{\let~\catcode~13 13\def^^M{^^J}~` 12\xdef\asciiart{

     88888888ba                                                                 
     88      "8b                                                                
     88      ,8P                                                                
     88aaaaaa8P'   ,adPPYba,    8b      db      d8    ,adPPYba,   8b,dPPYba,    
     88""""""'    a8"     "8a   `8b    d88b    d8'   a8P_____88   88P'   "Y8    
     88           8b       d8    `8b  d8'`8b  d8'    8PP"""""""   88            
     88           "8a,   ,a8"     `8bd8'  `8bd8'     "8b,   ,aa   88            
     88            `"YbbdP"'        YP      YP        `"Ybbd8"'   88

     ,ad8888ba,                                                            88   
    d8"'    `"8b                                        ,d                 88   
   d8'                                                  88                 88   
   88             8b,dPPYba,  8b       d8  ,adPPYba,  MM88MMM  ,adPPYYba,  88   
   88             88P'   "Y8  `8b     d8'  I8[    ""    88     ""     `Y8  88   
   Y8,            88           `8b   d8'    `"Y8ba,     88     ,adPPPPP88  88   
    Y8a.    .a8P  88            `8b,d8'    aa    ]8I    88,    88,    ,88  88   
     `"Y8888Y"'   88              Y88'     `"YbbdP"'    "Y888  `"8bbdP"Y8  88   
                                  d8'                                           
                                 d8'

}}\makeatletter\newif\iflabor%
\labortrue

\documentclass{amsart}

\usepackage{mathtools,amssymb}\allowdisplaybreaks

\usepackage[english]{babel}

\usepackage[utf8]{inputenc}
	\let\@DUC\DeclareUnicodeCharacter\@DUC{4EE4}{\@DO\@LET}
	\def\@DO#1{\bgroup\def\UTFviii@defined##1{\expandafter#1\string##1+}}
	\def\@LET#1:#2+{\egroup\@DUC{\UTFviii@hexnumber{\decode@UTFviii#2\relax}}}
	令定{\@DO\@DEF}\def\@DEF#1:#2+{\@LET:#2+{\@nameuse{ME#2AN}}\@namedef{ME#2AN}}
	
	令Γ\Gamma	令Δ\Delta	令Θ\Theta	令Λ\Lambda	令Ξ\Xi		令Π\Pi		
	令Σ\Sigma	令Υ\Upsilon	令Φ\Phi		令Ψ\Psi		令Ω\Omega	
	令α\alpha	令β\beta		令γ\gamma	令δ\delta	令ε\varepsilon			
	令ζ\zeta		令η\eta		令θ\theta	令ι\iota		令κ\kappa	令λ\lambda	
	令μ\mu		令ν\nu		令ξ\xi		令π\pi		令ρ\rho		令ς\varsigma	
	令σ\sigma	令τ\tau		令υ\upsilon	令φ\varphi	令χ\chi		令ψ\psi		
	令ω\omega	令ϑ\vartheta	令ϕ\phi		令ϖ\varpi	令Ϝ\Digamma 令ϝ\digamma	
	令κ\varkappa	令ϱ\varrho	令ϴ\varTheta	令ϵ\epsilon	
	令𝔽{\mathbb F}	令ℚ{\mathbb Q}	令ℤ{\mathbb Z}
	令𝒦{\mathcal K}	令ℋ{\mathcal H}	
	\DeclareMathAlphabet\mathsi{T1}\sfdefault\mddefault\sldefault
	令𝘈{\mathsi A}	令𝘊{\mathsi C}	令𝘚{\mathsi S}	令𝘙{\mathsi R}	
	令𝘢{\mathsi a}	令𝘴{\mathsi s}	
	令𝘞{\mathsi W}	
	令ℓ\ell		令∂\partial	令∇\nabla	
	令√\sqrt		
	\def\bigol#1{\bigl#1\iflabor\else\color{UCO}\fi}
	\def\bigor#1{\iflabor\else\color{UIB}\fi\bigr#1}
	\def\({\bigol(}	\def\){\bigor)}		令（{\Bigl(}		令）{\Bigr)}		
	令［{\bigol[}	令］{\bigor]}		令「{\Bigl[}		令」{\Bigr]}		
	令｛{\bigol\{}	令｝{\bigor\}}		令『{\Bigl\{}	令』{\Bigr\}}	
	令⌈\lceil	令⌉\rceil	令⌊\lfloor	令⌋\rfloor
	令⟨\langle	令⟩\rangle	令〈{\bigol\langle}		令〉{\bigor\rangle}
	\def\|{\mathrel\Vert}	令‖{\mathrel\Big\Vert}	令｜{\mid\nobreak}		
	令；{\mathrel;\nobreak}	令、{\setminus\nobreak}	令：{\colon}				
	令ˆ\hat		令˜\tilde	令¯\bar		令˘\breve	令˙\dot		令¨\ddot		
	令°\ocirc	令ˇ{^\vee}	
	令∏\prod		令∑\sum		令∫\int		令⋀\bigwedge	令⋁\bigvee	令⋂\bigcap	
	令⋃\bigcup	令⨁\bigoplus			令⨂\bigotimes			
	令±\pm		令·\cdot		令×\times	令÷\frac		令•\bullet	令∧\wedge	
	令∨\vee		令∩\cap		令∪\cup		令⊕\oplus	令⊗\otimes	令⊙\odot
	令⋆\star	
	令¬\neg		令∃\exists	令∞\infty	令⊤\top		令⊥\bot		令⋯\cdots	
	令♠\spadesuit			令♡\heartsuit			令♢\diamondsuit			
	令♣\clubsuit	令♭\flat		令♮\natural	令♯\sharp	令⋮\vdots
	令←\gets		令→\to		令↞\twoheadleftarrow		令↠\twoheadrightarrow	
	令↤\mapsfrom	令↦\mapsto	令↩\hookleftarrow		令↪\hookrightarrow		
	令↾{\mathbin\upharpoonright}			令∈\in		令∉\notin	令∋\ni		
	令≅\cong		令≈\approx	令≔\coloneqq	令≠\neq		令≡\equiv	令≤\leqslant	
	令≥\geqslant	令⊆\subseteq	令⊇\supseteq	令⟂\perp		令⟵\longleftarrow		
	令⟶\longrightarrow		令⟼\longmapsto			
	令…{,\allowbreak\dotsc,\allowbreak}
	定†#1†{\text{#1}}
	令©{\cellcolor{Dark yellow}}
	定®#1®{(#1-a_n)#1}
	令🔑{^⋆_}
	\def\[{\@ifstar{\begin{equation*}}{\begin{equation}}}
	\def\]{\@ifstar{\end  {equation*}}{\end  {equation}}}
	
	\DeclarePairedDelimiter\abs\lvert\rvert
	
	\DeclareMathOperator\spa{span}	\DeclareMathOperator\cha{char}	
	\def\ce{_{\text{ce}}}			\def\co{_{\text{co}}}
	\def\bma#1{\begin{bmatrix}#1\end{bmatrix}}
	
	\def\biy#1#2{\Bigl(\hbox{\smaller\!$\genfrac..\z@0{#1}{#2}$\!}\Bigr)}
	\def\biz#1#2{\Bigl(\hbox{\smaller[2]\!$\genfrac..\z@0{#1}{#2}$\!}\Bigr)}
	\def\bi#1#2{\mathchoice
		{\biy{#1}{#2}}{\binom{#1}{#2}}{\binom{#1}{#2}}{\binom{#1}{#2}}}
	\def\SM/{sym\-met\-ric-mul\-ti\-pli\-ca\-tion}
	\def\WM/{wedge-mul\-ti\-pli\-ca\-tion}
	\def\regen/{re\-gen\-er\-at\-ing}
	\def\MBR/{M\kern-.2ex\lower.5ex\hbox{B}\kern-.2exR}
	\def\MSR/{M\kern-.2ex\lower.5ex\hbox{S}\kern-.2exR}
	\def\MDS/{M\kern-.2ex\lower.5ex\hbox{D}\kern-.2exS}
	\def\subpack/{sub-packetization}

\usepackage{tikz,pgfplotstable,booktabs,colortbl}
	\usetikzlibrary{calc,shapes.geometric}
			\let\PMS\pgfmathsetmacro
	\let\PMT\pgfmathtruncatemacro	
	定色#1!#2 {\definecolor{#1}{HTML}{#2}}
	色UIB!13294b			色2728 C!0455A4		色2738 C!1F4096		
	色427!E8E9EA			色Cool Gray 6!A5A8AA	色Cool Gray 1!5E6669	
	色UCO!E84A27			色UIC red!D50032		色UIS blue!003366	
	色Teal!0d605e		色Gray-blue!6fafc7	色Citron!bfd46d		
	色Dark yellow!ffd125	色Salmon!ee5e5e		色Periwinkle!4f6898	
	\tikzset{every picture/.style={cap=round,join=round}}
	\pgfplotsset{compat/show suggested version=false,compat=1.13} 
	定點(#1)[#2]{node(d#1)[circle,fill,inner sep=1]{}node(D#1)[anchor=#2]}
	定籤(#1)--+(#2:#3){(#1.#2)--+(#2:#3)node[anchor=#2+180]}
	定苯#1(#2){node#1(#2)
		[regular polygon,regular polygon sides=6,inner sep=5,draw]{}}
	\pgfmathsetseed{\day+31*(\month+12*(\year-2000))}
	定彈#1{\advance#1\glueexpr0ptplus1ptminus1pt}

\usepackage[unicode,pdfusetitle,
	linkbordercolor=Salmon,citebordercolor=Teal,urlbordercolor=Gray-blue]{hyperref}
	\def\PM#1$#2${\texorpdfstring{$#2$}{#1}}
	\def\PT#1†#2†{\texorpdfstring{#2}{#1}}
	\def\U#1+{\unichar{"#1}}

\usepackage[capitalize,noabbrev]{cleveref}
	
	定名#1:#2?#3?#4 {\crefname{#1}{#2#3}{#2#4}}		
	名page:page??s			名section:section??s		名enumi:item??s
	\newtheorem{thm}{Theorem}名thm:Theorem??s
	定理#1:#2?#3?#4 {\newtheorem{#1}[thm]{#2#3}名#1:#2?#3?#4 }令風\theoremstyle
	理cor:Corollar?y?ies		理lem:Lemma??s			理pro:Proposition??s		
							理con:Conjecture??s		
	風{definition}			理dfn:Definition??s		理axi:Axiom??s			
							理alg:Algorithm??s		
	風{remark}				理cla:Claim??s			理rem:Remark??s			
	定式#1:#2?#3?#4 {名#1:#2?#3?#4 \creflabelformat{#1}{\textup{(##2##1##3)}}}
	式equ:equalit?y?ies		式sub:containment??s		式dia:diagram??s			
	式for:formula??s			式fun:function??s		式ine:inequalit?y?ies	
	式spa:space??s			式ten:tensor??s			
	\def\@ReadTypeLabel#1:#2?{\xdef\@TYPE{#1}\xdef\@LABEL{#2}}
	\def\eqlabel#1{\@ReadTypeLabel#1?\label[\@TYPE]{\@TYPE:\@LABEL}}
	\def\steplabel{\incr@eqnum\tag\theequation\eqlabel}
	\newcommand\taglabel[2][0]{\@ReadTypeLabel#2?\label[\@TYPE]{\@TYPE:\@LABEL}
		\addtocounter{equation}{#1}\tag{\theequation.\@LABEL}}
			
	\newcommand\mdslabel[2][]{\hypertarget{mds#1#2}{}\item[(\MDS/#2)]}
	\newcommand\mdsref[2][]{\textup{(\hyperlink{mds#1#2}{\MDS/#2})}}

定行{\catcode13 }行13定讀#1{行13\def^^M{\\}\def\@ATTRI{行5#1}\@ifnextchar[{選}{參}}
定選[^^M#1^^M]^^M#2^^M^^M{\@ATTRI[#1]{#2}}定參^^M#1^^M^^M{\@ATTRI{#1}}行5

讀
\title[
		   Multilinear Algebra for Minimum Storage Regenerating Codes   		
]
				            Multilinear Algebra for             				
				       Minimum Storage Regenerating Codes       				

讀
\author
				          Iwan Duursma\and Hsin-Po Wang         				

讀{\def
\@pdfsubject}
				         94B27; 15A69; math.IT; math.AC         				

讀
\subjclass
				         Primary 94B27; Secondary 15A69         				

讀
\thanks
		   This work was partially supported by NSF grant CCF-1619189.  		

\address{
				           Department of Mathematics,           				
				  University of Illinois at Urbana--Champaign,  				
				             Urbana, Illinois 61801             				
}
\email{%
				        duursma and hpwang2 @illinois.edu        				
}

\begin{document}\message{\asciiart}\makeatother

\begin{abstract}
	An $(n, k, d, \alpha)$-\MSR/ (minimum storage regeneration) code
	is a set of $n$ nodes used to store a file.
	For a file of total size $k\alpha$,
	each node stores $\alpha$ symbols, any $k$ nodes recover the file, and
	any $d$ nodes can repair any other node
	via each sending out $\alpha/(d-k+1)$ symbols.
	
	In this work, we explore various ways to re-express
	the infamous product-matrix construction using skew-symmetric matrices,
	polynomials, symmetric algebras, and exterior algebras.
	We then introduce a multilinear algebra foundation to produce
	$\bigl(n, k, \frac{(k-1)t}{t-1}, \binom{k-1}{t-1}\bigr)$-\MSR/ codes
	for general $t≥2$.
	At the $t=2$ end,
	they include the product-matrix construction as a special case.
	At the $t=k$ end, we recover determinant codes of mode $m=k$;
	further restriction to $n=k+1$ makes it identical to
	the layered code at the \MSR/ point.
	Our codes' \subpack/ level---$\alpha$---is independent of~$n$ and small.
	It is less than $L^{2.8(d-k+1)}$,
	where $L$ is Alrabiah--Guruswami's lower bound on~$\alpha$.
	Furthermore, it is less than other \MSR/ codes' $\alpha$
	for a subset of practical parameters.
	We offer hints on how our code repairs multiple failures at once.
\end{abstract}

彈\baselineskip/8 彈\lineskip/8 彈\parskip/4 彈\floatsep*2 彈\textfloatsep*2
\hbadness99\overfullrule1em\leftmargini4.5em

\maketitle

讀
\section
								  Introduction  								

	Distributed storage systems emerge as a nontraditional coding problem
	where the user gains and loses by multiples of
	a chunk of symbols called \emph{node}.
	The user wants to decode the original message
	by connecting to (only) a fraction of nodes.
	Moreover, nodes are actively checking for failures
	and are restored when a node failure is detected.
	This motivates the following definition:
	
	\begin{dfn}\label{dfn:regenerate}
		\cite{DGWWR10,WD09,RSKR09}
		An \emph{$(n,k,d,α,β,M)$-\regen/ code}
		is a collection of $n$ nodes used to store an $M$-symbol file.
		The storage is configured such that
		(a)	each node stores $α$ symbols;
		(b)	any $k$ nodes contain sufficient information
			to recover the file; and
		(c)	any $d$ nodes can repair any other failing node
			by each sending out $β$ symbols.
	\end{dfn}
	
	In terms of random variables and entropies
	\cite[(4)--(6)]{Duursma14} \cite[Definition~1]{Tian14},
	a file $Φ$ is a (random) vector in $𝔽^M$,
	where $𝔽$ is the working alphabet.
	Each node stores a vector $𝘞_h∈𝔽^α$ depending on $Φ$,
	where $h∈[n]≔\{1,2…n\}$ is the node index.
	That means $H(𝘞_h｜Φ)=0$ for all $h∈[n]$.
	Any $k$ vectors (any $k$ nodes) suffice to recover the file $Φ$, so
	\[*H(Φ｜𝘞_{h_1},𝘞_{h_2}…𝘞_{h_k})=0\]*
	for arbitrary distinct indices $h_1,h_2…h_k∈[n]$.
	The actual procedure that recovers $Φ$ from $𝘞_{h_1},𝘞_{h_2}…𝘞_{h_k}$ is
	called the \emph{downloading scheme} or the \emph{data recovery scenario}.
	
	When, say, the $f$th node fails for some $f∈[n]$,
	a subset $ℋ⊆[n]、\{f\}$ of $d$ nodes will be asked to help.
	A helper node with index $h∈ℋ$ sends a vector
	$𝘚^ℋ_{h→f}∈𝔽^β$ to repair the failing one.
	That means $H(𝘚^ℋ_{h→f}｜𝘞_h)=0$
	for all $f∈[n]$ and all $h∈ℋ⊆[n]、\{f\}$.
	The content of the failing node can be derived from the help messages.
	To rephrase it,
	\[*H(𝘞_f｜𝘚^ℋ_{h_1→f},𝘚^ℋ_{h_2→f}…𝘚^ℋ_{h_d→f})=0\]*
	for arbitrary distinct indices $f,h_1,h_2…h_d∈[n]$ and $ℋ≔\{h_1,h_2…h_d\}$.
	The actual procedure that recovers $𝘞_f$
	from $𝘚^ℋ_{h_1→f},𝘚^ℋ_{h_2→f}…𝘚^ℋ_{h_d→f}$ is called
	the \emph{repairing scheme} or the \emph{node repairing scenario}.
	
	This definition immediately poses a dilemma:
	In order to store files more efficiently,
	node contents should share very little mutual information.
	But then, repairing a node becomes more difficult as it is hard
	to find relations among vectors sharing little mutual information.
	The quantity $β$ is referred to as the \emph{repair bandwidth}
	as it represents the required bandwidth of the network
	(from a helper to the failure).
	Another interpretation is that, when the code is linear,
	$dβ/α$ is the average length of the parity check equations
	used to compute symbols in the $f$th node.
	
	From here researches split into two paths.
	The first path characterizes
	the homogeneous trade-off among $α$, $β$, and~$M$.
	Here, ratios $α/M$ and $β/M$ are used
	to measure the normalized node size and bandwidth, respectively.
	An illustrative trade-off between $α/M$ and $β/M$
	is plotted in \cref{fig:433}.
	It has $(k,d)=(3,3)$ and arbitrary $n≥4$.
	The inner bound and the outer bound meet in this case, i.e.,
	existing codes achieve the theoretically best trade-off.
	In general, however, the two bounds disagree;
	more works are needed to close the gap.
	For the latest results on the achievable side,
	see \cite{
		RSK11o, 
		SRKR12d,SRKR12i, 
		TSAVK15, 
		SSK15, 
		GEC14, 
		EM16d,EM16d,EM19c,
		DL19} 
	and references therein.
	See \cite{Duursma14,
		PK15, 
		Tian15,SPKVSK16, 
		EMT15, 
		MT15, 
		LL16, 
		Duursma19} 
	for the latest results on the unfeasible side.
	Together they summarize existing works on the first path.
	
	\begin{figure}
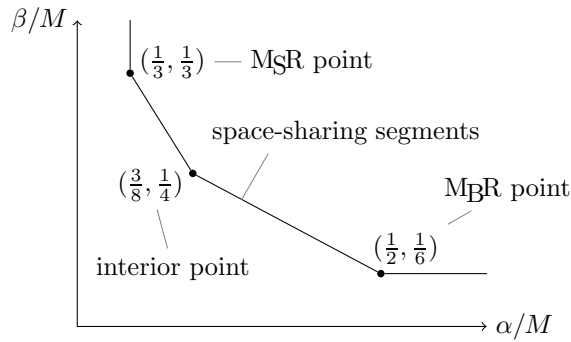

		$$\tikz[x=20cm,y=16cm]{
			\draw[<->]
				(1/3,1/3)	點(1)[195]{$(÷13,÷13)$}
				(3/8,1/4)	點(2)[15]{$(÷38,÷14)$}
				(1/2,1/6)	點(4)[225]{$(÷12,÷16)$}
				($(d1)+(-2em,2em)$)coordinate(Y)node[left]{$β/M$}|-
				($(d4)+(4em,-2em)$)coordinate(X)node[right]{$α/M$};
			\iflabor
			\draw(Y-|d1)--(d1)--(d2)--coordinate[pos=.25](e3)(d4)--(d4-|X);
			\draw[help lines,nodes=black]
				籤(D1)--+(0:1em){\MSR/ point}
				籤(e3)--+(60:2em)[anchor=180+15]{space-sharing segments}
				籤(D4)--+(30:1em){\MBR/ point}
				籤(D2)--+(-75:1.5em){interior point};
			\fi
		}$$
		\caption{
			The $α/M$-to-$β/M$ trade-off when
			$(k,d)=(3,3)$ and $n≥d+1=4$ is arbitrary.
			It is settled in the sense that:
			(a)	every solid point is achieved
				by some existing regenerating codes;
			(b)	every segment is achieved by the \emph{space-sharing} technique;
			and	(c)	every point outside (exclusively)
			the segments is provably not achievable
			owing to some carefully-crafted inequalities of Shannon type.
			(Neither axis starts from~$0$.)
		}\label{fig:433}
	\end{figure}
	
	In a trade-off plot such as \cref{fig:433},
	the lower right solid point is called the \emph{\MBR/}
	(minimum bandwidth regeneration) point since it minimizes $β/M$.
	The upper left solid point is called the \emph{\MSR/}
	(minimum storage regeneration) point because it minimizes $α/M$.
	Both \MBR/ and \MSR/ points are of particular interest
	for their extremity as well as the fact that
	existing codes achieve the cut-set bound for all parameters.
	Between the two, the \MSR/ point attracts notable attention
	as it strengthens the \MDS/ property through
	asking for the optimal repair bandwidth.
	Regenerating codes aiming for the \MSR/ point
	are what constitute the second path.
	On this path, \cref{dfn:regenerate} takes a simpler form.
	
	\begin{dfn}\label{dfn:msr}
		\cite{WD09,RSKR09}
		An $(n,k,d,α,α/(d-k+1),kα)$-regenerating code
		is called an \emph{$(n,k,d,α)$-\MSR/ code}.
		The parameter $α$ is called the \emph{\subpack/ level}.
	\end{dfn}
	
	Remark:
	Historically, an $(n,k,d,α)$-\MSR/ code is first an $[n,k]$-\MDS/ code
	over $𝔽^α$ and then equipped with the repairing property.
	Why $β=α/(d-k+1)$ is the least possible repair bandwidth when $kα=M$
	is not hard to see.
	Here we adapt the argument of information flow from \cite{DGWWR10}.
	
	Say we want the file and download the first $k-1$ nodes $𝘞_1,𝘞_2…𝘞_{k-1}$.
	Instead of downloading one more node, we pretend that
	the $n$th node fails and ask the first $d$ nodes to repair.
	We know the help messages from the first $k-1$ nodes
	$𝘚^{[d]}_{1→n},𝘚^{[d]}_{2→n}…𝘚^{[d]}_{k-1→n}$ because they can be derived
	from the node contents $𝘞_1,𝘞_2…𝘞_{k-1}$, respectively.
	What is new are the other help messages
	$𝘚^{[d]}_{k→n},𝘚^{[d]}_{k+1→n}…𝘚^{[d]}_{d→n}$.
	Together we have $(k-1)α+(d-k+1)β$ symbols.
	From here we can reconstruct the $n$th node.
	Since we now have the full contents of $k$ nodes, we comprehend the file~$Φ$.
	In virtue of the conservation law of information, $(k-1)α+(d-k+1)β≥M$.
	Since $M$ is fixed to be $kα$, we obtain $β≥α/(d-k+1)$.
	This type of argument is what \emph{cut-set bounds} refer to.
	
	Having the cut-set bounds in mind, works aiming at the \MSR/ point either
	stick to $(β,M)≔(α/(d-k+1),kα)$ or, less frequently, require proximity.
	It is then reasonable to ask,
	What is the minimal \subpack/ level $α$ a code can achieve?
	A series of works \cite{GTC14,BK18t,AG19} pursue the answer from below;
	the best known lower bound on $α$ is the following.
	
	\begin{thm}\label{thm:lower}
		\cite[Theorem~1]{AG19}
		For any $(n,k,d,α)$-\MSR/ code,
		\[*α≥\exp（÷{k-1}{4(d-k+1)}）.\]*
		(Remark:
			It is later discovered that their bound is not valid when $k=d$.
			See \cref{app:benchmarks} for more details.)
	\end{thm}
	
	Other works pursued the minimal $α$ through inventing new codes.
	During this period, some appealing properties are defined and fulfilled.
	One instance is to fix $n=d+1$ and ask for the so-called
	\emph{optimal-access} property, where every helper disk reads and
	transfers $β=α/(d-k+1)$ symbols without forming linear combinations.
	The best result in this paradigm is clay code \cite{VRPKLSKBYNHN18}.
	Another instance is to relax the restriction $n=d+1$ and show that $n-d$
	can be arbitrarily large at the expense of increasing $α$ \cite{RKV16}.
	Yet another branch is to refine Reed--Solomon codes over large fields;
	refinement here means that a help message is not
	a whole symbol in the big field, but a fraction of it.
	See \cite{CYB20} for the latest update
	on which of the aforementioned nice properties about
	the repairment of Reed--Solomon codes are enabled.
	There are also works that focus on
	repairing multiple failing nodes at once.
	This is further branched into two models---one model allows failing nodes
	to help each other while the other prohibits \cite{CJMRS13,SH13,YB19}.
	Lastly, we remark that some works proposed that since the exactness
	in $α=β(d-k+1)$ creates too much burden ($α$ being exponential in $n$ etc.),
	one considers relaxing it ``by $ϵ$'' \cite{GLJ18}.
	In doing so, the \subpack/ level $α$ grows logarithmically in~$n$.
	This save is, colloquially, doubly-exponential in $n$.
	
	In this paper, we fall back to the classical \cref{dfn:msr},
	where nodes fail one at a time, no access property is considered,
	and the overall code is not Reed--Solomon in itself.
	We first review the well-known product-matrix construction.
	The product-matrix code, originated from \cite{RSK11o}, paved the path of
	\MSR/ codes and accumulates a decent amount of interests,
	for both simplicity and a \subpack/ level as low as $α=k-1$.
	Despite of the popularity, we have not encountered any code
	that specializes to product-matrix code.
	
	Later when we were working on \cite{DLW20},
	we found that multilinear algebra is
	the right language to describe certain regenerating codes.
	We attempt and succeed in describing product-matrix
	compactly in terms of multilinear algebra.
	We present this description after a brief algebra review.
	The description further leads to a natural extension of
	product-matrix codes, which is the main contribution of this paper.
	
	\begin{thm}[main theorem]\label{thm:main}
		Let $n$, $k$, $d$, and $t$ be integers
		such that $n-1≥d≥k≥2$ and $d≤t(d-k+1)$.
		Let $α≔\bi{(t-1)(d-k+1)}{t-1}$.
		If $α≤3003$, then there exists an $(n,k,d,α)$-\MSR/ code
		over some sufficiently large field.
	\end{thm}
	
	We name it \emph{Atrahasis code} after the fictional character
	who survived a seven-day flood in an Akkadian epic recorded on clay tablets.
	
	Two proofs of the main theorem are found in \cref{sec:atrahasis}.
	As will be clarified later, both proofs depend on
	whether a certain determinant is non-vanishing.
	We precomputed all cases under $α≤3003$, and found no counterexample.
	We believe that this determinant is nonzero for all $α$.
	
	\begin{con}\label{con:all}
		\Cref{thm:main} holds for all $α$.
	\end{con}

\subsection{Paradigms comparison}\label{sec:paradigm}

	A comparison is made in \cref{tab:paradigm}.
	From top to bottom:
	product-matrix at the \MBR/ point \cite[section~IV]{RSK11o};
	and then at the \MSR/ point [ibid, section~V];
	clay code family \cite{SAK15,YB17o,VRPKLSKBYNHN18};
	attempts of \cite{GFV17,RKV16} to separate $n$ from $d+1$;
	this work extending the product-matrix approach;
	refinement of Reed--Solomon codes \cite{GW17,TYB19,CYB20};
	$ϵ$-\MSR/ code relaxing the cut-set bound \cite{RTGE17e,GLJ18};
	layered code \cite{TSAVK15};
	determinant code \cite{EM19d};
	and cascade code \cite{EM19c}.
	From left to right, whether the code:
	achieves the \MSR/ point;
	aims for points between \MSR/ and \MBR/;
	achieves the \MBR/ point;
	achieves the cut-set bound;
	allows $d>k$ besides the $d=k$ case;
	allows $n>d+1$ besides the $n=d+1$ case;
	and has the optimal-access property.
	The last two columns list
	the expected \subpack/ level $α$ and the working field size $\abs{𝔽}$.
	``Alon'' means that the only general bound on field size
	comes from the combinatorial Nullstellensatz \cite{Alon99}.
	See \cref{sec:field} for detailed bounds on and instances of $\abs{𝔽}$.
	
	\pgfplotstableread[header=false]{
		code       \MSR/ --- \MBR/ cut $d>k$ $n>d+1$ I/O α≈         $\abs{𝔽}≈$
		prod-mat@B   X    X    O    O    O      O     X  d          $n$       
		prod-mat@S   O    X    X    O    O      O     X  k-1        $n$       
		clay         O    X    X    O    O      X     O  r^{n/r}    $n$       
		GFV17        O    X    X    O    O      O     X  \rknkr/    Alon      
		RKV16        O    X    X    O    O      O     X  r^{n/r}    Alon      
		Atrahasis    O    X    X    O    O      O     X  \bikt/     Alon      
		refined-RS   O    X    X    O    O      O     O  n^n        $n-1$     
		{$ϵ$-\MSR/}  O    X    X   $ϵ$   O      O     X  \log n     $O(n)$    
		layered      O    O    O    O    X      X     O  \tb k{k/2} $2$       
		determinant  O    O    O    O    X      O     X  \tb k{k/2} $n$       
		cascade      O    O    O    X    O      O     X  r^k        $n$       
	}\pgfParadigm
	\begin{table}
		\caption{
			A comparison about what parameters each paradigm is interested in;
			$r=d-k+1$ and $t=⌈d/(d-k+1)⌉$.
			See \cref{sec:paradigm} for details.
			See \crefrange{fig:lowerbound}{fig:EM19grid} starting from
			\cpageref{fig:lowerbound} for details about $α$.
		}\label{tab:paradigm}
		\newbox\BIKT\setbox\BIKT=\hbox{$\biy{k-1}{t-1}$\vphantom{\bigg|}}
		\def\bikt/{\box\BIKT}	\def\rknkr/{r^{k\binom{n-k}r}}	\def\tb{\tbinom}
		\def\arraystretch{1.44}
		\centering\pgfplotstabletypeset[
			every head row/.style=output empty row,
			every row no 0/.style={before row=\toprule,after row=\midrule},
			every odd row/.style={before row=\rowcolor{427}},
			every last row/.style={after row=\bottomrule},
			string replace=O{$\bigcirc$},string replace=X{$×$},
			string type,
			columns/8/.style={assign cell content/.style={
				@cell content/.initial=$\displaystyle##1$}},
		]\pgfParadigm
	\end{table}

\subsection{Shortening fills gaps}

	Throughout existing works, it is common to see that
	the code construction is given for a sparse family of parameters,
	but that does not mean the code only applies to a small range of situations.
	This is because there is a way to tune the parameters of an \MSR/ code.
	More precisely, we have a lemma.
	
	\begin{lem}[shortening an \MSR/-code]\label{lem:shorten}
		Given an $(n,k,d,α)$-\MSR/ code,
		that is, an $(n,k,d,α,α/(d-k+1),kα)$-\regen/ code,
		there exists an $(n-1,k-1,d-1,α)$-\MSR/ code,
		that is, an $(n-1,k-1,d-1,α,α/(d-k+1),(k-1)α)$-\regen/ code,
		over the same alphabet.
	\end{lem}
	
	\begin{proof}
		The key idea is to constrain that
		the $n$th node stores constant contents.
		For instance, let $0∈𝔽$ be a symbol in the working alphabet.
		Then we set $𝘞_n=\bma{0&0&⋯&0}∈𝔽^α$.
		
		For number of nodes ($n$):
		Since we don't need any storage to keep an all-zero vector $𝘞_n$,
		we retire the $n$th node.
		Now there are $n-1$ nodes left.
		
		For the node size ($α$):
		Since the first $n-1$ nodes stores what they used to,
		the node size remains the same;
		the old $α$ is the new $α$.
		
		For the file size ($M$):
		Consider the encoding functions of the last $k$ nodes
		(including the $n$th) as a whole $𝘞_{n-k+1}^n：𝔽^M→(𝔽^α)^k$.
		Since $M=kα$ in the old \MSR/ code, $𝘞_{n-k+1}^n$ is a bijection.
		Since we then fix $𝘞_n$, the file can only take values in the preimage
		\[*\{Φ∈𝔽^M:†last $α$ components of †𝘞_{n-k+1}^n(Φ)=𝘞_n\}⊆𝔽^M.\]*
		So the preimage is of cardinality $\abs{𝔽}^{(k-1)α}$.
		This leads to the new file size $(k-1)α$,
		which is the ``new~$k$'' multiplied by the ``new $α$''.
		
		For downloading scheme ($k$):
		We want that any $k-1$ from the first $n-1$ nodes recover the file.
		This is possible because whenever we download $k-1$ nodes,
		we remember setting the $n$th node all-zero.
		This means that
		we know the content of $k$ nodes in the old \MSR/ code.
		By definition, any $k$ nodes recover the file in the old \MSR/ code.
		So any $k-1$ nodes recover the file in the new \MSR/ code.
		
		For repair bandwidth ($β$):
		The remaining nodes execute the repairing scheme as usual,
		so $β=α/(d-k+1)$ remains the same,
		which is also the ``new $α$'' divided by the ``new $(d-k+1)$''.
		
		For number of helpers ($d$): 
		Since all nodes know $𝘞_n$, any failing node will
		ask for $d-1$ helpers and simulate how the $n$th node could have helped.
		Since this means that the failing node has $d$ help messages
		(one derived from $𝘞_n$), it can repair itself.
		So $d-1$ is the new $d$.
	\end{proof}
	
	This technique is called \emph{shortening}
	as it mimics the shortening of linear block codes.
	It bears the same meaning as in the title of \cite{Duursma19}.
	The technique can be applied iteratively.
	
	\begin{lem}[shortening satuaration]\label{lem:shortclose}
		Let $n,k,d,α,δ$ be positive integers.
		Given an $(n+δ,k+δ,d+δ,α)$-\MSR/ code.
		There exists an $(n,k,d,α)$-\MSR/ code over the same alphabet.
	\end{lem}
	
	Note that $d-k+1$ is invariant under successive shortening.
	The main functionality of shortening is to reduce our main theorem
	to a task of composing a sparse family of \MSR/ codes.
	More precisely,
	the following theorem and \cref{lem:shortclose} imply \cref{thm:main}.
	
	\begin{thm}[primitive step]\label{thm:primitive}
		Fix integers $n$, $k$, and $d$ such that $n-1≥d≥k≥2$.
		Assume $t≔d/(d-k+1)$ is an integer and $α=\bi{t(d-k+1)}{t-1}≤3003$.
		Then there exists an $(n,k,d,α)$-\MSR/ code
		over some sufficiently large field.
	\end{thm}
	
	\Cref{sec:atrahasis} in its entirety serves as
	the proof of \cref{thm:primitive} modulo the field size part.
	\Cref{sec:field} completes the field size part.

\subsection{Organization}

	\Cref{sec:prodmat} reviews the product-matrix code at the \MSR/ point.
	\Cref{sec:algebra} prepares some algebra definitions
		for our paraphrase and generalization of product-matrix.
	\Cref{sec:pmalgebra} paraphrases the product-matrix framework
		in terms of multilinear algebra.
	\Cref{sec:highrate} states an explicit $(9,5,6,6)$-\MSR/ code
		and then moves on to $(n,k,3(k-1)/2,α)$-\MSR/ codes
		as a nontrivial example and a bridge to the general result.
	\Cref{sec:atrahasis} proves \cref{thm:primitive} modulo the field size part.
	\Cref{sec:field} handles the field size part.
	\Cref{app:multiple} analyzes the performance of Atrahasis code
		when two nodes fail at once.
	\Cref{app:benchmarks} compares existing codes numerically.

讀
\section
					    Product-Matrix at \PT MSR†\MSR/†    					

\label{sec:prodmat}

	In this section, we review the classical idea of product-matrix
	at the \MSR/ point \cite[section~V]{RSK11o}.
	The construction consists of two parts:
	a $d=2(k-1)$ \MSR/ code and a stretching to $d>2(k-1)$.
	The precise statement of the former is below.
	
	\begin{thm}[primitive product-matrix]\label{thm:pmorig}
		\cite[section~V]{RSK11o}
		Let $n-1≥d=2(k-1)≥2$.
		There exists an $(n,k,d,k-1)$-\MSR/ code
		over any field $𝔽$ such that $\abs{\{a^{k-1}:a∈𝔽\}}≥n$.
	\end{thm}
	
	This, with \cref{lem:shortclose}, immediately implies the following.
	
	\begin{pro}[stretched product-matrix]
		\cite[section~V.C]{RSK11o}
		Let $n-1≥d≥2(k-1)≥2$.
		There exists an $(n,k,d,d-k+1)$-\MSR/ code
		over any field $𝔽$ such that $\abs{\{a^{k-1}:a∈𝔽\}}≥n+d-2(k-1)$.
	\end{pro}
	
	We brief the proof of \cref{thm:pmorig} in the rest of this section.
	How to generalize product-matrix to $d<2(k-1)$ cases remains open
	since \cite{RSK11o} was published.
	This region is usually referred to as
	the \emph{high-rate region} in literature.
	Our main contribution in \cref{sec:atrahasis}
	answers the question positively.

\subsection{The primitive construction}\label{sec:pmsym}

	Assume $n-1≥d=2(k-1)≥2$.
	Hereby we recite the $(n,k,d,k-1)$-\MSR/ code construction
	of product-matrix.
	To specify this and every other code construction, we go over four steps:
	file format and $M$ (closely related to $n$), node configuration and $α$,
	downloading scheme (closely related to $k$),
	and repairing scheme and $β$ (closely related to $d$).
	Follow the subsubsection titles.

\subsubsection{File format and \PM M$M$}

	Let $𝔽$ be a field of order $n$ or greater.
	Over $𝔽$, let
	\[*𝘚_1,𝘚_2∈𝔽^{(k-1)×(k-1)}\tag{file format}\]*
	be two $(k-1)$-by-$(k-1)$ symmetric matrices.
	We use $(𝘚_1,𝘚_2)$ to pre-encode the file.
	That is to say, since each symmetric matrix has $(k-1)k/2$ free entries,
	they jointly represent a file of size $M≔(k-1)k$ symbols.

\subsubsection{Node configuration and $α$}

	For each $h∈[n]$, the $h$th node selects a scalar $ξ_h∈𝔽$
	and a (row) vector $y🔑h∈𝔽^{k-1}$.
	The node then stores the vector
	\[*y🔑h𝘚_1+ξ_hy🔑h𝘚_2∈𝔽^{k-1}.\tag{node content}\]*
	That means, each node stores $α≔k-1$ symbols.
	For downloading and repairing,
	we put some requirements on the selection of $ξ_h$s and $y🔑h$s.
	
	\begin{axi}\label{axi:pmsym}
		The selection of $ξ_h$s and $y🔑h$s shall meet
		the following three \MDS/ requirements.
		\begin{itemize}
			\mdslabel{x}	All $ξ_h$ are distinct.
			\mdslabel{y}	Any $k-1$ many $y🔑h$s span $𝔽^{k-1}$.
				That is, $\spa⟨y🔑{h_1},y🔑{h_2}…y🔑{h_{k-1}}⟩=𝔽^{k-1}$
				for all distinct indices  $h_1,h_2…h_{k-1}∈[n]$.
			\mdslabel{d}	Any $d$ concatenated vectors
				$\bma{y🔑h&ξ_hy🔑h}$ span $𝔽^d$.
				That is to say,	\\
				$\spa〈\bma{y🔑1&ξ_1y🔑1},
					\bma{y🔑2&ξ_2y🔑2}…\bma{y🔑d&ξ_dy🔑d}〉=𝔽^d$
				for all distinct indices $h_1,h_2…h_d∈[n]$.
		\end{itemize}
	\end{axi}
	
	\Cref{sec:pmselect} breaks down how to find $ξ_h$s and $y🔑h$s
	based on Reed--Solomon codes.

\subsubsection{Downloading scheme}

	We now explain why any $k$ nodes recover
	the file in the format of $(𝘚_1,𝘚_2)$.
	In doing so, observe that $y🔑1𝘚_1(y🔑2)^⊤$ behaves like
	a bi-linear form in $y🔑1$ and $y🔑2$.
	Furthermore, it is symmetric because $𝘚_1$ is%
	---$y🔑1𝘚_1(y🔑2)^⊤=y🔑2𝘚_1(y🔑1)^⊤$.
	
	\begin{pro}
		Let $𝘚_1$, $𝘚_2$ be symmetric but unknown.
		Let $ξ_h$s and $y🔑h$s satisfy \mdsref{x} and \mdsref{y}.
		Then k many $y🔑h𝘚_1+ξ_hy🔑h𝘚_2$ uniquely determine $𝘚_1,𝘚_2$.
	\end{pro}
	
	\begin{proof}
		Due to the symmetry possessed by the node configuration, it suffices
		to check if the first $k$ nodes recover the file $(𝘚_1,𝘚_2)$.
		Fix any distinct indices $i,j∈[k]$.
		We download the vector $y🔑i𝘚_1+ξ_iy🔑i𝘚_2$,
		so we can deduce the scalar $(y🔑i𝘚_1+ξ_iy🔑i𝘚_2)(y🔑j)^⊤$,
		which happens to be $y🔑i𝘚_1(y🔑j)^⊤+ξ_iy🔑i𝘚_2(y🔑j)^⊤$.
		Similarly, we download the vector $y🔑j𝘚_1+ξ_jy🔑j𝘚_2$,
		so we can deduce the scalar $(y🔑j𝘚_1+ξ_jy🔑j𝘚_2)(y🔑i)^⊤$,
		which happens to be $y🔑i𝘚_1(y🔑j)^⊤+ξ_jy🔑i𝘚_2(y🔑j)^⊤$ by symmetry.
		Hence we can now \emph{decouple} the values
		\[*
			\bma{
				y🔑i𝘚_1(y🔑j)^⊤+ξ_iy🔑i𝘚_2(y🔑j)^⊤\\
				y🔑i𝘚_1(y🔑j)^⊤+ξ_jy🔑i𝘚_2(y🔑j)^⊤
			}=\bma{
				1&ξ_i \\ 1&ξ_j
			}\bma{
				y🔑i𝘚_1(y🔑j)^⊤ \\ y🔑i𝘚_2(y🔑j)^⊤
			}.
		\]*
		The square matrix above is invertible
		because \mdsref{x} reads $ξ_i≠ξ_j$ .
		So we can deduce (separate/isolate) the value of $y🔑i𝘚_1(y🔑j)^⊤$.
		This leads to an oracle that outputs the value of $y🔑i𝘚_1(y🔑j)^⊤$
		for any distinct $i,j∈[k]$.
		We now call the oracle for a fixed $i$ and arbitrary $j∈[k]、\{i\}$.
		Owing to \mdsref{y}, $y🔑j$ for $j∈[k]、\{i\}$ span $𝔽^{k-1}$,
		so we can recover $y🔑i𝘚_1$ as a vector.
		Now we vary $i$, and conclude that we can recover $𝘚_1$ as a matrix.
		For $𝘚_2$, repeat the same procedure after
		getting the decoupled value $y🔑i𝘚_2(y🔑j)^⊤$.
		This procedure that recovers both $𝘚_1$ and $𝘚_2$ witnesses
		the claim that a file can be recovered from any $k$ nodes.
	\end{proof}

\subsubsection{Repairing scheme and $β$}

	Let $f∈[n]$ be the index of a failing node.
	Let $ℋ⊆[n]、\{f\}$ be the $d$ helper nodes
	that are going to transmit help messages.
	For every $h∈ℋ$, the $h$th node will transmit
	\[*(y🔑h𝘚_1+ξ_hy🔑h𝘚_2)(y🔑f)^⊤∈𝔽\tag{help message}\]*
	to the $f$th node.
	The left parentheses enclose the content of the $h$th node.
	This message is a $1$-by-$1$ scalar so $β≔1$.
	Now we verify that the failing node can repair its content
	after receiving $d$ many help messages.
	
	\begin{pro}
		Let $𝘚_1$, $𝘚_2$ be symmetric but unknown.
		Let $ξ_h$s and $y🔑h$s satisfy \mdsref{d}.
		Then $d$ many $(y🔑h𝘚_1+ξ_hy🔑h𝘚_2)(y🔑f)^⊤$
		uniquely determine $y🔑f𝘚_1+ξ_fy🔑f𝘚_2$.
	\end{pro}
	
	\begin{proof}
		Without loss of generality,
		assume that the first $d$ nodes are helping and that $f>d$.
		Then what the failing node receives can be rewritten as
		\[*(y🔑h𝘚_1+ξ_hy🔑h𝘚_2)(y🔑f)^⊤=
			\bma{y🔑h&ξ_hy🔑h}\bma{𝘚_1(y🔑f)^⊤ \\ 𝘚_2(y🔑f)^⊤}\]*
		for $h∈[d]$.
		The right-hand side is the product of
		a $1$-by-$d$ vector with a $d$-by-$1$ vector (recall $d=2(k-1)$).
		\mdsref{d} reads that $\bma{y🔑1&ξ_hy🔑1}…\bma{y🔑d&ξ_hy🔑d}$
		span $𝔽^d$--- i.e., they form an invertible matrix.
		Hence the failing node can reproduce $𝘚_1(y🔑f)^⊤$ and $𝘚_2(y🔑f)^⊤$.
		Now it remains to compute the linear combination
		$𝘚_1(y🔑f)^⊤+𝘚_2(y🔑f)^⊤ξ_f=(y🔑f𝘚_1+ξ_fy🔑f𝘚_2)^⊤$
		in order to restore the content $y🔑f𝘚_1+ξ_fy🔑f𝘚_2$.
	\end{proof}
	
	This concludes the $(n,k,2(k-1),k-1)$-\MSR/ code specification
	needed to prove \cref{thm:pmorig}, modulo field size.
	Before addressing field size in \cref{sec:pmselect},
	we offer an alternative construction for the same region of parameters.

\subsection{The skew construction}\label{sec:pmext}

	One straightforward variant of the previous subsection is that
	the symmetric matrices $𝘚_1,𝘚_2$ of dimensions $(k-1)×(k-1)$
	can be replaced by skew-symmetric matrices $𝘈_1,𝘈_2$ of dimensions $k×k$.
	Firstly we observe that this does not change the file size;
	it is still $M=k(k-1)$.
	Next we lengthen ``$y🔑h$'' such that
	``$y🔑h𝘈_1$'' and other products make sense.
	Although it now seems like the node size should be $k$, we claim that
	we can still form an \MSR/ code with the exact same parameters as before.
	That is, an $(n,k,d,k-1)$-\MSR/ code for all $n-1≥d=2(k-1)≥2$.
	We elaborate the specification in the rest of this subsection.

\subsubsection{File format and \PM M$M$}

	Let $𝔽$ be a field of order $n$ or greater.
	Over $𝔽$, let
	\[*𝘈_1,𝘈_2∈𝔽^{k×k}\tag{file format}\]*
	be two $k$-by-$k$ skew-symmetric matrices when $\cha𝔽≠2$.
	When $\cha𝔽=2$, let they have zeros on the diagonal.
	They are matrices such that $w𝘈_1w^⊤=w𝘈_2w^⊤=0$ for all $w∈𝔽^k$.
	We use $(𝘈_1,𝘈_2)$ to pre-encode the file.
	Since each matrix has $k(k-1)/2$ free entries,
	they jointly represent a file of size $M=k(k-1)$ symbols.

\subsubsection{Node configuration and $α$}

	For each $h∈[n]$, let the $h$th node select
	a scalar $ξ_h∈𝔽$ and a nonzero (row) vector $w🔑h∈𝔽^k$.
	Then it stores
	\[*w🔑h𝘈_1+ξ_hw🔑h𝘈_2∈𝔽^k.\tag{node content}\]*
	It looks like the node needs to store $k$ symbols
	but $k-1$ symbols suffice.
	This is because $(w🔑h𝘈_1+ξ_hw🔑h𝘈_2)(w🔑h)^⊤=0$---%
	the vector to be stored lies in a codimension-$1$ subspace.
	We now have $α≔k-1$.
	For downloading and repairing to work, we assign some requirements
	on the selection of $ξ_h$s and $w🔑h$s (cf. \cref{axi:pmsym}).
	
	\begin{axi}\label{axi:pmext}
		The selection of $ξ_h$s and $w🔑h$s shall comply with
		the following three \MDS/ requirements:
		\begin{itemize}
			\mdslabel[']{x}	All $ξ_h$ are distinct.
				(Same as in \cref{axi:pmsym}.)
			\mdslabel{w}	Any $k$ many $w🔑h$s span $𝔽^k$.
				That is, $\spa⟨w🔑{h_1}…w🔑{h_k}⟩=𝔽^k$
				for all distinct $h_1…h_k∈[n]$.
				(Dimension changed accordingly.)
			\mdslabel{q}	Any $d$ concatenated vectors $\bma{w🔑h&ξ_hw🔑h}$
				span $𝔽^{2k}/〈\bma{w🔑f&0},\bma{0&w🔑f}〉$.
				That is, $\spa〈\bma{w🔑1&ξ_1w🔑1}…\bma{w🔑d&ξ_dw🔑d},
				\bma{w🔑f&0},\bma{0&w🔑f}〉=𝔽^{2k}$
				for all distinct $f,h_1…h_d∈[n]$.
				(Dimension changed accordingly.)
		\end{itemize}
	\end{axi}
	
	\Cref{sec:pmselect} deals with how to find $ξ_h$s and $w🔑h$s.

\subsubsection{Downloading scheme}

	Notice that $w🔑1𝘈_1(w🔑2)^⊤=-w🔑2𝘈_1(w🔑1)^⊤$
	and, in particular, $w🔑1𝘈_1(w🔑1)^⊤=0$.
	We are to verify that any $k$ node contents recover the file $(𝘈_1,𝘈_2)$.
	
	\begin{pro}
		Let $𝘈_1$, $𝘈_2$ be skew-symmetric matrices
		with zero diagonal and with unknown elements off the diagonal.
		Let $ξ_h$s and $w🔑h$s satisfy \mdsref[']{x} and \mdsref{y}.
		Then k many $w🔑h𝘈_1+ξ_hw🔑h𝘈_2$ uniquely determine $𝘈_1,𝘈_2$.
	\end{pro}
	
	\begin{proof}
		On account of the symmetry, it suffices to demonstrate
		how to recover the file from the first $k$ nodes.
		Fix any distinct indices $i,j∈[n]$.
		We deduce the scalar
		$(w🔑i𝘈_1+ξ_iw🔑i𝘈_2)(w🔑j)^⊤=w🔑i𝘈_1(w🔑j)^⊤+ξ_iw🔑i𝘈_2(w🔑j)^⊤$
		from what we download from the $i$th node.
		We deduce the scalar
		$(w🔑j𝘈_1+ξ_jw🔑j𝘈_2)(w🔑i)^⊤=-w🔑i𝘈_1(w🔑j)^⊤-ξ_jw🔑i𝘈_2(w🔑j)^⊤$
		from what is downloaded form the $j$th node.
		Now decouple.
		\[*
			\bma{
				w🔑i𝘈_1(w🔑j)^⊤+ξ_iw🔑i𝘈_2(w🔑j)^⊤\\
				-w🔑i𝘈_1(w🔑j)^⊤-ξ_jw🔑i𝘈_2(w🔑j)^⊤
			}=\bma{
				1&ξ_i \\ -1&-ξ_j
			}\bma{
				w🔑i𝘈_1(w🔑j)^⊤ \\ w🔑i𝘈_2(w🔑j)^⊤
			}
		\]*
		By \mdsref[']{x}, the square matrix is invertible.
		Hence we can isolate the value of $w🔑i𝘈_1(w🔑j)^⊤$.
		We now have $w🔑i𝘈_1(w🔑j)^⊤$ for
		a fixed $i$ and various $j∈[k]、\{i\}$.
		Besides, we know $w🔑i𝘈_1(w🔑i)^⊤$ (which is $0$).
		On grounds of \mdsref{w},
		we have collected the products of $w🔑i𝘈_1$ with a basis of $𝔽^k$,
		which leads to the recovery of $w🔑i𝘈_1$ as a vector.
		Then we vary $i$ to rebuild $𝘈_1$ as a matrix.
		For $𝘈_2$, repeat the same procedure with $w🔑i𝘈_2(w🔑j)^⊤$.
	\end{proof}

\subsubsection{Repairing scheme and $β$}

	Let $f$ be the index of a failing node.
	Let $ℋ⊆[n]、\{f\}$ be the $d$ helper nodes that will transmit help messages.
	For every $h∈ℋ$, the $h$th node transmits
	\[*(w🔑h𝘚_1+ξ_hw🔑h𝘚_2)(w🔑f)^⊤∈𝔽\tag{help message}\]*
	to the $f$th node.
	This is a $1$-by-$1$ scalar so $β≔1$.
	Next, we justify that the failing node can repair
	after gathering $d$ help messages.
	
	\begin{pro}
		Let $𝘈_1$, $𝘈_2$ be skew-symmetric matrices
		with zero diagonal and with unknown entries off the diagonal.
		Let $ξ_h$s and $w🔑h$s satisfy \mdsref{q}.
		Then $d$ many $(w🔑h𝘈_1+ξ_hw🔑h𝘈_2)\*(w🔑f)^⊤$ (granted that $h≠f$)
		uniquely determine $w🔑f𝘈_1+ξ_fw🔑f𝘈_2$.
	\end{pro}
	
	\begin{proof}
		By virtue of the symmetry, we assume $ℋ=[d]$ and $f>d$.
		Then the failing node rewrites what it receives:
		\[*(w🔑h𝘈_1+ξ_hw🔑h𝘈_2)(w🔑f)^⊤=
			\bma{w🔑h&ξ_hw🔑h}\bma{𝘈_1(w🔑f)^⊤ \\ 𝘈_2(w🔑f)^⊤}\]*
		for all $h∈[d]$.
		Other than that, the $f$th node knows
		\[*0=\bma{0&w🔑f}\bma{𝘈_1(w🔑f)^⊤ \\ 𝘈_2(w🔑f)^⊤}
			=\bma{w🔑f&0}\bma{𝘈_1(w🔑f)^⊤ \\ 𝘈_2(w🔑f)^⊤}\]*
		as part of the code construction.
		So it knows the product of
		\[*\bma{𝘈_1(w🔑f)^⊤ \\ 𝘈_2(w🔑f)^⊤}∈𝔽^{2k×1}\]*
		with vectors $\bma{w🔑1&ξ_hw🔑1}…\bma{w🔑d&ξ_hw🔑d}$,
		$\bma{w🔑f&0}$, and $\bma{w🔑f&0}$.
		Those vectors span $𝔽^{2k}$ by \mdsref{q}, so the failing node
		can infer $𝘈_1(w🔑f)^⊤$ (the transpose of $w🔑f𝘈_1$)
		and $𝘈_2(w🔑f)^⊤$ (the transpose of $w🔑f𝘈_2$).
		Thus it infers the original node content $w🔑f𝘈_1+ξ_fw🔑f𝘈_2$.
	\end{proof}
	
	This concludes the alternative $(n,k,2(k-1),k-1)$-\MSR/ code construction.
	Next we address how to select $ξ_h$s, $y🔑h$s, and $w🔑h$s.

\subsection{Selecting \PM\U03BE+$ξ$, \PM y$y$, and \PM w$w$}
	\label{sec:pmselect}

	\cite{RSK11o} suggested using Reed--Solomon codes.
	Here are the details.
	
	\begin{lem}
		Let $a_1,a_2…a_n∈𝔽$ be such that
		$a_1^{k-1},a_2^{k-1}…a_n^{k-1}$ are all distinct.
		For each $h∈[n]$,
		let $ξ_h≔a_h^{k-1}∈𝔽$ and $y🔑h≔\bma{1&a_h&⋯&a_h^{k-2}}∈𝔽^{k-1}$.
		Then \cref{axi:pmsym} is satisfied.
	\end{lem}
	
	\begin{proof}
		(Remark:
			Since $\abs{\{a^{k-1}:a∈𝔽\}}≥n$,
			the existence of $a_1,a_2…a_n$ is a non-problem.)
		First, \mdsref{x} is satisfied because
		the $(k-1)$th powers of the points are all distinct.
		Next, \mdsref{y} is satisfied because $y🔑h$s
		are (transposes of) distinct column vectors of a Reed--Solomon code;
		and Reed--Solomon codes are \MDS/ codes.
		Lastly, \mdsref{d} is satisfied because
		$\bma{y🔑h&ξ_hy🔑h}=\bma{1&a_h&⋯&a_h^{d-1}}$
		is again a column of a Reed--Solomon code.
	\end{proof}
	
	This concludes the field size part of \cref{thm:pmorig}.
	A similar idea is used to fulfill \cref{axi:pmext}.
	
	\begin{lem}
		Let $a_1,a_2…a_n∈𝔽$ be such that
		$a_1^{k-1},a_2^{k-1}…a_n^{k-1}$ are all distinct.
		For each $h∈[n]$,
		let $ξ_h≔a_h^{k-1}∈𝔽$ and $w🔑h≔\bma{1&a_h&⋯&a_h^{k-1}}∈𝔽^k$.
		Then \cref{axi:pmext} holds.
	\end{lem}
	
	\begin{proof}
		\mdsref[']{x} and \mdsref{w} hold for the same reason
		\mdsref{x} and \mdsref{y} in the previous lemma do.
		For \mdsref{q}, it suffices to check that this $2k$-by-$2k$ matrix
		\[*\bma{
			1			&1			&⋯	&1			&1			&			\\
			a_1			&a_2		&⋯	&a_d		&©a_n		&			\\
			⋮			&⋮			&	&⋮			&©⋮			&			\\
			a_1^{k-1}	&a_2^{k-1}	&⋯	&a_d^{k-1}	&©a_n^{k-1}	&			\\
			\vphantom{\vrule height1.44em}
			a_1^{k-1}	&a_2^{k-1}	&⋯	&a_d^{k-1}	&			&1			\\
			a_1^k		&a_2^k		&⋯	&a_d^k		&			&©a_n		\\
			⋮			&⋮			&	&⋮			&			&©⋮			\\
			a_1^d		&a_2^d		&⋯	&a_d^d		&			&©a_n^{k-1}	\\
		}\]*
		is invertible.
		To do so, we attempt to eliminate shaded entries using row operations.
		For each $i=k,k-1…2$ (notice the order),
		subtract $a_n$ times the $(i-1)$th row from the $i$th row.
		We arrive at:
		\[*\bma{
			1			&1			&⋯	&1			&1			&			\\
			®a_1®		&®a_2®		&⋯	&®a_d®		&			&			\\
			⋮			&⋮			&	&⋮			&			&			\\
			®a_1®^{k-1}	&®a_2®^{k-1}&⋯	&®a_d®^{k-1}&			&			\\
			\vphantom{\vrule height1.44em}
			a_1^{k-1}	&a_2^{k-1}	&⋯	&a_d^{k-1}	&			&1			\\
			a_1^k		&a_2^k		&⋯	&a_d^k		&			&©a_n		\\
			⋮			&⋮			&	&⋮			&			&©⋮			\\
			a_1^d		&a_2^d		&⋯	&a_d^d		&			&©a_n^{k-1}	\\
		}\]*
		For each $i=k,k-1…2$,
		subtract $a_n$ times the $(k+i-1)$th row from the $(k+i)$th row.
		We reach:
		\[*\bma{
			1			&1			&⋯	&1			&1			&			\\
			®a_1®		&®a_2®		&⋯	&®a_d®		&			&			\\
			⋮			&⋮			&	&⋮			&			&			\\
			®a_1®^{k-1}	&®a_2®^{k-1}&⋯	&®a_d®^{k-1}&			&			\\
			\vphantom{\vrule height1.44em}
			a_1^{k-1}	&a_2^{k-1}	&⋯	&a_d^{k-1}	&			&1			\\
			®a_1®^k		&®a_2®^k	&⋯	&®a_d®^k	&			&			\\
			⋮			&⋮			&	&⋮			&			&			\\
			®a_1®^d		&®a_2®^d	&⋯	&®a_d®^d	&			&			\\
		}\]*
		Eliminate the first and the $(k+1)$th rows using the last two columns.
		Rescale all but the last two columns.
		Then we are left with a Vandermonde minor.
	\end{proof}
	
	What we were doing here looks like---and in fact is---%
	shortening a Reed--Solomon code to from a generalized Reed--Solomon code.
	We knew the matrix is invertible because the latter code is \MDS/.

\subsection{A polynomial shorthand}\label{sec:pmpoly}

	As Reed--Solomon codes admit polynomial descriptions,
	so do codes built upon Reed--Solomon codes.
	Here is a concise paraphrase of the primitive construction
	paired with Reed--Solomon vectors in terms of polynomials.

\subsubsection{File format and \PM M$M$}

	Let $𝔽[y,y']_{k-2}$ be the set of
	symmetric polynomials of bi-degree $(k-2,k-2)$ or less.
	To put it another way, $𝔽[y,y']_{k-2}$ is a vector space over $𝔽$ spanned by
	$1$, $y+y'$, $yy'$, $y^2+y'^2$, $y^2y'+yy'^2$, $y^3+y'^3…y^{k-2}y'^{k-2}$.
	One can identify the coefficient of $y^{i-1}y'^{j-1}$ with
	the $(i,j)$th entry of a $(k-1)$-by-$(k-1)$ symmetric matrix.
	Let $𝘴_1(y,y'),𝘴_2(y,y')∈𝔽[y,y']_{k-2}$.
	Then the coefficients of $𝘴_1(y,y'),𝘴_2(y,y')$
	carry a file of size $M=k(k-1).$

\subsubsection{Node configuration and $α$}

	For each $h∈[n]$, the $h$th node stores
	\[*𝘴_1(a_h,y')+a_h^{k-1}𝘴_2(a_h,y')∈𝔽[y']\tag{node content}\]*
	as a polynomial in $y'$.
	This univariate polynomial has degree $k-2$ or less, so $α=k-1$.

\subsubsection{Downloading scheme}

	Say we download the first $k$ nodes.
	Fix distinct $i,j∈[k]$.
	We can specialize $𝘴_1(a_i,y')+a_i^{k-1}𝘴_2(a_i,y')$
	to $𝘴_1(a_i,a_j)+a_i^{k-1}𝘴_2(a_i,a_j)∈𝔽$.
	So can we specialize $𝘴_1(a_j,y')+a_j^{k-1}𝘴_2(a_j,y')$
	to $𝘴_1(a_j,a_i)+a_j^{k-1}𝘴_2(a_j,a_i)=𝘴_1(a_i,a_j)+a_j^{k-1}𝘴_2(a_i,a_j)∈𝔽$.
	Now we possess two evaluations of the polynomial
	$𝘴_1(a_i,a_j)+x𝘴_2(a_i,a_j)∈𝔽[x]$, at $x=a_i^{k-1}$ and at $x=a_j^{k-1}$.
	Therefore, we can recover the constant term $𝘴_1(a_i,a_j)$
	and the linear term $𝘴_2(a_i,a_j)$.
	Repeat this for all $i≠j$, then we can recover $𝘴_1$ and $𝘴_2$
	as we have sufficiently many evaluations.

\subsubsection{Repairing scheme and $β$}

	When the $f$th node fails, the $h$th node sends
	\[*𝘴_1(a_h,a_f)+a_h^{k-1}𝘴_2(a_h,a_f)∈𝔽\tag{help message}\]*
	to the $f$th node for every $h∈ℋ$.
	This is a field element so $β=1$.
	
	Now consider $𝘴_1(y,a_f)+y^{k-1}𝘴_2(y,a_f)∈𝔽[y]$
	as a polynomial in $y$ of degree $d$ or less.
	Then the help messages are
	evaluations of this polynomials at $d$ distinct points.
	Therefore, the failing node can learn $𝘴_1(y,a_f)$ (the lower degree part)
	and $𝘴_2(y,a_f)$ (the higher degree part).
	And it determines $𝘴_1(y,a_f)+a_f^{k-1}𝘴_2(y,a_f)$.
	
	We end this section with a remark that
	a similar description can be carried out with anti-symmetric polynomials.

讀
\section
							   Algebra Background   							

\label{sec:algebra}

	This section gives self-contained definitions of
	tensor, symmetric, and exterior algebras
	that will be used in our construction.
	Contents of this section can be found in standard textbooks.
	To skip, proceed to \cref{sec:pmalgebra} on \cpageref{sec:pmalgebra}.
	
	Let $𝔽$ be a field.
	Our framework measures information in $𝔽$-symbols
	so the finiteness of $𝔽$ is not mandatory.
	However, finite fields---especially those with characteristic~$2$---%
	are usually assumed for applications (distributed storage).
	On the other hand, a crucial part of the construction
	implies that the field must have sufficiently many elements;
	we elaborate the implication later in \cref{sec:field}.
	
	Let $U,V,W$ be finite dimensional vector spaces over $𝔽$.
	Elements of $U$ are denoted by $u$ with or without proper subscripts,
	elements of $V$ by $v$, and element of $W$ by $w$.
	For brevity, we call vector spaces \emph{spaces}.
	
	The \emph{dual space} of $U$, denoted by $Uˇ$,
	is the space consisting of all linear transformations from $U$ to $𝔽$.
	We call elements of $Uˇ$ \emph{functionals} to distinguish them
	from elements of $U$, which we call \emph{vectors}.
	Since $U$ is of finite dimension, $U$ and $Uˇ$ share the same dimension.
	Furthermore, $(Uˇ)ˇ$ is isomorphic to $U$ canonically---a vector $u∈U$ gives
	rise to a map from $Uˇ$ to $𝔽$ by mapping a functional $ϕ∈Uˇ$ to $ϕ(u)∈𝔽$.
	It turns out that linear transformations defined in this way
	exhaust all possible linear transformations from $Uˇ$ to $𝔽$.
	The field element $ϕ(u)∈𝔽$ is called the \emph{evaluation of $ϕ$ at $u$}.
	The action that takes a functional $ϕ∈Uˇ$ as the input and returns $ϕ(u)∈𝔽$
	is called \emph{evaluating $ϕ$ at $u$} or simply \emph{evaluating at $u$}.
	For any subspace $V⊆U$, the \emph{restriction of $ϕ$ to $V$}
	is a functional from $V$ to $𝔽$ that evaluates $v∈V⊆U$ to $ϕ(v)$.
	This restriction is denoted by $ϕ↾V$.
	The corresponding action is called \emph{restricting $ϕ$ to $V$},
	or simply \emph{restricting to $V$}.
	
	A crucial part of our construction involves evaluations of
	a functional $ϕ∈Uˇ$ at a list of vectors $u_1,u_2,u_3,\dotsc∈U$.
	Interesting things happens when these vectors share some linear relations.
	For instance, if we want to evaluate $ϕ∈Uˇ$ at $u_1$, $u_2$, and $u_1-3u_2$,
	then we can also evaluate at only the first two vectors $u_1,u_2$ and
	compute the third evaluation by linearity $ϕ(u_1-3u_2)=ϕ(u_1)-3ϕ(u_2)$.
	From an information theoretic perspective,
	the information content of $ϕ(u_1)$, $ϕ(u_2)$, and $ϕ(u_1-3u_2)$
	is no more than that of $ϕ(u_1)$ and $ϕ(u_2)$.
	More generally, if $V$ is a subspace of $U$ and
	we want to know the restriction $ϕ↾V$, it suffices to choose
	a basis of $V$ (any basis) and evaluate at each vector in the basis.
	For all intents and purposes,
	which basis is used does not affect the properties of the codes;
	only the size of the basis $\dim V$ matters.

\subsection{Tensors and tensor products}

	Let $¯u_1,¯u_2…¯u_d∈U$ form a basis of $U$ of dimension $d$.
	Let $¯v_1,¯v_2…¯v_l∈V$ form a basis of $V$ of dimension $l$.
	Denoted by $U⊗V$, the \emph{tensor product} of $U$ and $V$
	is the space that consists of formal sums of the form
	\[∑_{ij}a_{ij}¯u_i⊗¯v_j.\eqlabel{for:tensorbasis}\]
	Here $a_{ij}∈𝔽$, and each $¯u_i⊗¯v_j$ is an unbreakable, free variable
	whose sole purpose is to carry its coefficient.
	The addition is term-wise:
	\[*∑_{ij}a_{ij}¯u_i⊗¯v_j+∑_{ij}b_{ij}¯u_i⊗¯v_j
		≔∑_{ij}(a_{ij}+b_{ij})¯u_i⊗¯v_j.\]*
	The scalar multiplication is distributive:
	\[*c·∑_{ij}a_{ij}¯u_i⊗¯v_j≔∑_{ij}(ca_{ij})¯u_i⊗¯v_j\]*
	for any $c∈𝔽$.
	The dimension is $\dim(U⊗V)=\dim(U)·\dim(V)=dl$.
	
	It is quite obvious that we could have put $a_{ij}$ into a $d$-by-$l$ array
	and define $U⊗V$ to be the space of arrays (matrices).
	However, doing so prevents us from seeing the greater picture:
	we may pretend that the character ``$⊗$'' is
	an infixed binary operator from $U⊕V$ to $U⊗V$ that sends
	\[*(u,v)=（∑_ia_i¯u_i,∑_jb_j¯v_j）∈U⊕V,\]*
	where $a_i,b_j∈𝔽$, to
	\[u⊗v≔∑_{ij}(a_ib_j)¯u_i⊗¯v_j∈U⊗V.\eqlabel{for:tensorany}\]
	This map is \emph{bi-linear} in the sense that it is linear in $u$, meaning
	\[*(u+cu')⊗v=∑_{ij}(a_ib_j+ca_i'b_j)¯u_i⊗¯v_j=u⊗v+cu'⊗v,\]*
	and linear in $v$, meaning
	\[*u⊗(v+cv')=∑_{ij}(a_ib_j+ca_ib_j')¯u_i⊗¯v_j=u⊗v+cu⊗v'.\]*
	But it is not linear in both,
	meaning $(u+cu')⊗(v+cv')≠u⊗v+cu'⊗v'$ in general.
	Once we give $u⊗v$---the juxtaposition of ``$⊗$'' with arbitrary vectors---%
	an interpretation, describing an element of $U⊗V$ can be done by
	summing a finite list of $u_i⊗v_i$ where these $u_i$ and $v_i$ are
	not necessarily the same vectors as $¯u_i$ and $¯v_i$.
	We then treat $U⊗V$ as a collection of formal sums of the form
	$∑_ia_iu_i⊗v_i$ subject to the bi-linearity relation,
	where $u_i∈U$ and $v_i∈V$ are arbitrary vectors.
	The addition of formal sums is done by adding the coefficients
	of the matched $(u_i⊗v_i)$-terms and leaving unmatched terms intact.
	For example $(2u_1⊗v_1+u_2⊗7v_2)$ plus $(-u_2⊗v_2+u_3⊗8v_3)$
	is equal to $(2u_1⊗v_1+6u_2⊗v_2+8u_3⊗v_3)$.
	This is the basis-free definition of $U⊗V$.
	A corollary is that no matter which basis we choose
	in \cref{for:tensorbasis} we will end up defining
	\emph{the} vector space structure on $U⊗V$, up to isomorphism.
	
	We call an element of $U⊗V$ a \emph{tensor} to distinguish it
	from vectors, elements of plainer spaces like $U,V,W$.
	The fact that $-u_1⊗v_1-u_2⊗v_2+u_1⊗v_2+u_2⊗v_1$ and
	$u_2⊗(-v_2+v_1)-u_1⊗(v_1-v_2)$ along with $(-u_1+u_2)⊗v_1+(u_1-u_2)⊗v_2$ as
	well as $(u_2-u_1)⊗(v_1-v_2)$ describe the same tensor inspires a question,
	What is the least amount of ``$⊗$'' required to describe a tensor?
	In the tensor product of only two spaces,
	this question boils down to decomposing a matrix $[a_{ij}]_{ij}$ into
	a product $𝘊𝘙$ of a $d$-by-$r$ matrix $𝘊$ and an $r$-by-$l$ matrix $𝘙$
	with the least possible $r$.
	(Remark:
		When $r$ reaches the minimum,
		columns of $𝘊$ are a basis of the column space of $[a_{ij}]_{ij}$;
		rows of $𝘙$ are a basis of the row space.)
	The number $r$ is called the \emph{rank of a tensor},
	which resembles the rank of a matrix.
	When $r=1$, the tensor is of the form $au⊗v$ for $a∈𝔽$ and $(u,v)∈U⊕V$.
	This is called a \emph{rank-$1$ tensor} or a \emph{simple tensor}.
	
	The new tensor notation defined in \cref{for:tensorany}
	possesses more convenience than \cref{for:tensorbasis}.
	Consider again the tensor product $U⊗V$.
	We interpret $u⊗V$ as the collection of tensors of the form $∑a_iu⊗v_i$,
	that is, the formal sums where the ``$U$-component'' is always $u$.
	We interpret $U⊗v$ as the collection of tensors of the form $∑a_iu_i⊗v$.
	If $W$ is a subspace of $U$, then we interpret $W⊗V$ as the collection
	of tensors where the ``$U$-component'' is always from $W$.
	It is easy to check that $u⊗V$, $U⊗v$, and $W⊗V$ are all subspaces of $U⊗V$.
	
	The tensor notation generalizes to combinations of three or more spaces.
	Let $U$ and $V$ have bases $¯u_1,¯u_2…¯u_d$
	and $¯v_1,¯v_2…¯v_l$, respectively.
	Let $W$ be a $k$-dimensional space with basis $¯w_1,¯w_2…¯w_k$.
	It is not hard to imagine that $U⊗(V⊗W)$, $(U⊗V)⊗W$, and
	any other similar combination all give the same vector space structure.
	It is common to unify them as $U⊗V⊗W$,
	a space consisting of formal sums of the form
	\[*∑_{hij}a_{hij}¯u_h⊗¯v_i⊗¯w_j.\]*
	The addition is term-wise.
	The scalar multiplication is distributive.
	The dimension is $\dim(U⊗V⊗W)=\dim(U)·\dim(V)·\dim(W)=dlk$.
	Similar to \cref{for:tensorany}, we interpret
	\[*u⊗v⊗w=（∑_ha_h¯u_h）⊗（∑_ib_i¯v_i）⊗（∑_jc_j¯w_j）,\]*
	where $(u,v,w)∈U⊕V⊕W$ and $a_h,b_i,c_j∈𝔽$, as
	\[*∑_{hij}(a_hb_ic_j)¯u_h⊗¯v_i⊗¯w_j∈U⊗V⊗W.\]*
	It is \emph{tri}-linear in the sense that
	$(u+cu')⊗v⊗w=u⊗v⊗w+cu'⊗v⊗w$ and $u⊗(v+cv')⊗w=u⊗v⊗w+cu⊗v'⊗w$
	along with $u⊗v⊗(w+cw')=u⊗v⊗w+cu⊗v⊗w'$.
	This again gives us a versatile way to describe tensors in $U⊗V⊗W$,
	namely by formal sums of the form
	\[*∑_ia_iu_i⊗v_i⊗w_i.\]*
	We can ask again what is the least possible length of formal sums
	that describe a certain tensor, and call this number its \emph{rank}.
	And then we can talk about whether a tensor is of rank one;
	a rank-$1$ tensor is of the form $au⊗v⊗w$.
	In a general tensor product of three or more spaces, computing the rank
	or determining whether a tensor is of rank one is difficult.
	But all we need is that every tensor is the sum of several
	rank-$1$ tensors, i.e., rank-$1$ tensors span the whole space.
	As a consequence, we can describe a linear transformation from
	a tensor space by describing the image of every rank-$1$ tensor.
	
	The dual of a tensor product is the tensor product of duals,
	i.e., $(U⊗V⊗W)ˇ$ is isomorphic to $Uˇ⊗Vˇ⊗Wˇ$.
	Let $ϕ∈(U⊗V⊗W)ˇ$ be a functional.
	(We do not have a word to distinguish ``plain'' functionals in
		$Uˇ$, $Vˇ$, or $Wˇ$ and tensor-flavored functionals in $(U⊗V⊗W)ˇ$.)
	Every tensor is a sum of rank-$1$ tensors and $ϕ$ is linear,
	so describing $ϕ$ is equivalent to
	describing $ϕ$'s evaluations at rank-$1$ tensors.

\subsection{Tensor power, symmetric power, and exterior power}

	Let $T^0V$ be $𝔽$;
	let $T^1V$ be $V$;
	and let $T^pV$ be a product $V⊗V⊗\dotsb⊗V$ of $p$ many $V$.
	This is called \emph{the $p$th tensor power of $V$}.
	Some authors write $V^{⊗p}$.
	Let $¯v_1,¯v_2…¯v_l$ form a basis of $V$.
	Tensors in $T^pV$ are of the form
	\[∑a_{i_1i_2\dotsm i_p}¯v_{i_1}⊗¯v_{i_2}⊗\dotsb⊗¯v_{i_p}
		\eqlabel{for:powerbasis}\]
	where $a_{i_1i_2\dotsm i_p}∈𝔽$ and
	the summation is over $i_1,i_2…i_p∈[l]$. Here $[l]≔\{1,2…l\}$.
	Same as before, we allow arbitrary vectors to build-up
	rank-$1$ tensors before summing them.
	Thus a tensor in $T^pV$ can be described by a sum of the form
	$∑_ia_iv_{i1}⊗v_{i2}⊗\dotsb⊗v_{ip}$
	where $a_i∈𝔽$ and $v_{ij}∈V$ are arbitrary vectors.
	The addition is done via matching rank-$1$ tensors.
	The scalar multiplication is distributive.
	The dimension is $\dim(T^pV)=\dim(V)^p=l^p$.
	To avoid confusion, it is worth noting that
	$v_1⊗v_2$ is in general \emph{not} equal to $v_2⊗v_1$
	unless $v_1$ is a multiple of $v_2$ or $v_2=0$.
	
	Let $Y$ be an $m$-dimensional space over $𝔽$.
	Let $S^0Y$ be $𝔽$;
	Let $S^1Y$ be $Y$.
	Let $¯y_1,¯y_2…¯y_m$ form a basis of $Y$.
	Let $S^qY$ be the space consisting of formal sums of the form
	\[*∑a_{i_1i_2\dotsm i_q}¯y_{i_1}⊙¯y_{i_2}⊙\dotsb⊙¯y_{i_q}\]*
	where the summation is over all $1≤i_1≤i_2≤\dotsb≤i_q≤m$ and
	each $¯y_{i_1}⊙¯y_{i_2}⊙\dotsb⊙¯y_{i_q}$ is an unbreakable, free variable.
	This is called \emph{the $q$th symmetric power of $Y$}.
	The addition is term-wise.
	The scalar multiplication is distributive.
	The dimension is $\dim(S^qY)=\bi{\dim(Y)+q-1}q=\bi{m+q-1}q$.
	When $m≤0$ or $q<0$, the summation is empty,
	so the space is a singleton $𝔽^0=\{0\}$.
	The space becomes interesting after we define the \emph{\SM/}
	\[*Θ：T^qY⟶S^qY\]*
	that sends $¯y_{j_1}⊗¯y_{j_2}⊗\dotsb⊗¯y_{j_q}$ to
	\[*¯y_{i_1}⊙¯y_{i_2}⊙\dotsb⊙¯y_{i_q}\]*
	where $i_1≤i_2≤\dotsb≤i_q$ is the sorted copy of indices $j_1,j_2…j_q$.
	For a sum of many $¯y_{j_1}⊗¯y_{j_2}⊗\dotsb⊗¯y_{j_q}$
	like \cref{for:powerbasis},
	$Θ$ applies to each summand and the images are added together.
	This makes $Θ$ a linear transformation.
	
	We call elements of $S^qY$ \emph{tensors}.
	The \SM/ $Θ$ allows us to describe tensors in $S^qY$ more concisely:
	We interpret
	\[*y_1⊙y_2⊙\dotsb⊙y_q\]*
	as
	\[*Θ(y_1⊗y_2⊗\dotsb⊗y_q)∈S^qY\]*
	where $y_1,y_2…y_q∈Y$.
	Then we can use arbitrary vectors in $Y$ to describe tensors in $S^qY$---%
	what make up $S^qY$ are formal sums of rank-$1$ tensors of the form
	$∑_ia_iy_{i1}⊙y_{i2}⊙\dotsb⊙y_{iq}$
	where $a_i∈𝔽$ and $y_{ij}∈Y$ are arbitrary vectors.
	The addition is done via matching rank-$1$ tensors.
	The scalar multiplication is distributive.
	This syntax has the following two infamous characterizations.
	\begin{description}
		\item[\kern-1emMultilinearity]
			it is linear in every ``$y$'', meaning that \\
			$y_1⊙\dotsb⊙(y_i+cy_i')⊙\dotsb⊙y_q$ is equal to \\
			$(y_1⊙\dotsb⊙y_i⊙\dotsb⊙y_q)+c(y_1⊙\dotsb⊙y_i'⊙\dotsb⊙y_q)$.
		\item[\kern-1emCommutativity]
			swapping two $y$'s does nothing, \\
			$y_1⊙\dotsb⊙y_i⊙\dotsb⊙y_j⊙\dotsb⊙y_q
				=y_1⊙\dotsb⊙y_j⊙\dotsb⊙y_i⊙\dotsb⊙y_q$.
	\end{description}
	The proof is routine and omitted.
	Note that tensors in $T^pV$ are also multilinear in the same sense---%
	$v_1⊗\dotsb⊗(v_i+cv_i')⊗\dotsb⊗v_q$ is equal to
	$(v_1⊗\dotsb⊗v_i⊗\dotsb⊗v_q)+c(v_1⊗\dotsb⊗v_i'⊗\dotsb⊗v_q)$.
	
	Consider the symmetric square $S^2Y$.
	We interpret $y⊙Y$ as the collection of tensors of the form $∑a_iy⊙y_i$,
	that is, the formal sums where the first component is always $y$.
	We interpret $Y⊙y$ as the collection of tensors of the form $∑a_iy_i⊙y$,
	which is the same subset as $y⊙Y$.
	For higher symmetric powers,
	one can interpret $y_1⊙y_2⊙Y$, $y_1⊙Y⊙y_2$, $Y⊙y_1⊙Y⊙y_2$, etc.\ similarly.
	It is easy to verify that
	they are all subspaces of $S^qY$ for the obvious choices of $q$.
	In particular, $Y⊙Y⊙\dotsb⊙Y=S^qY$.
	
	Let $Λ^0W$ be $𝔽$;
	let $Λ^1W$ be $W$.
	Let $¯w_1,¯w_2…¯w_k$ form a basis of $W$.
	Let $Λ^qW$ be the space consisting of formal sums of the form
	\[*∑a_{i_1i_2\dotsm i_q}¯w_{i_1}∧¯w_{i_2}∧\dotsb∧¯w_{i_q}\]*
	where the summation is over all $1≤i_1<i_2<\dotsb<i_q≤k$ and
	each $¯w_{i_1}∧¯w_{i_2}∧\dotsb∧¯w_{i_q}$ is an unbreakable, free variable.
	This is called \emph{the $q$th exterior power of $W$}.
	The addition is term-wise.
	The scalar multiplication is distributive.
	The dimension is $\bi kq$.
	When $q<0$ or $q>k$, the summation is empty,
	so the space is a singleton $𝔽^0=\{0\}$.
	The space becomes interesting after we define the \emph{\WM/}
	\[*Δ：T^qW⟶Λ^qW\]*
	that sends $¯w_{j_1}⊗¯w_{j_2}⊗\dotsb⊗¯w_{j_q}$ to
	\[*\begin{cases*}
		0∈Λ^qW & if some indices coincide, \\
		(-1)^σ¯w_{i_1}∧¯w_{i_2}∧\dotsb∧¯w_{i_q} & otherwise,
	\end{cases*}\]*
	where $i_1<i_2<\dotsb<i_q$ is the sorted copy of indices $j_1,j_2…j_q$,
	and $σ$ is the number of swaps used to sort.
	The parity of $σ$ is commonly called the \emph{parity} of the permutation
	that sends $i_1$ to $j_1$, sends $i_2$ to $j_2$, et seq.
	For a sum of many $¯w_{j_1}⊗¯w_{j_2}⊗\dotsb⊗¯w_{j_q}$
	like \cref{for:powerbasis},
	$Δ$ applies to each summand and the images are added together.
	This makes $Δ$ a linear transformation.
	
	Elements of $Λ^qW$ are sometimes called \emph{multi-vectors}.
	We still call them \emph{tensors}.
	The \WM/ $Δ$ allows us to describe tensors in $Λ^qW$ more concisely:
	We interpret
	\[*w_1∧w_2∧\dotsb∧w_q\]*
	as
	\[*Δ(w_1⊗w_2⊗\dotsb⊗w_q)∈Λ^qW\]*
	where $w_1,w_2…w_q∈W$.
	Then we can use arbitrary vectors in $W$ to describe tensors in $Λ^qW$---%
	what make up $Λ^qW$ are formal sums of rank-$1$ tensors of the form
	$∑_ia_iw_{i1}∧w_{i2}∧\dotsb∧w_{iq}$
	where $a_i∈𝔽$ and $w_{ij}∈W$ are arbitrary vectors.
	The addition is done via matching rank-$1$ tensors.
	The scalar multiplication is distributive.
	The dimension is $\dim(Λ^qW)=\bi{\dim(W)}q=\bi kq$.
	This syntax has the following two infamous characterizations.
	\begin{description}
		\item[\kern-1emMultilinearity]
			It is linear in every ``$w$'', meaning that \\
			$w_1∧\dotsb∧(w_i+cw_i')∧\dotsb∧w_q$ is equal to \\
			$(w_1∧\dotsb∧w_i∧\dotsb∧w_q)+c(w_1∧\dotsb∧w_i'∧\dotsb∧w_q)$.
		\item[\kern-1emAnti-commutativity]
			$w_1∧\dotsb∧w_i∧\dotsb∧w_j∧\dotsb∧w_q=0$ if $w_i=w_j$. \\
			This implies that swapping two $w$'s causes a sign change, \\
			$w_1∧\dotsb∧w_i∧\dotsb∧w_j∧\dotsb∧w_q
				=-w_1∧\dotsb∧w_j∧\dotsb∧w_i∧\dotsb∧w_q$.
	\end{description}
	The proof is routine and omitted.
	
	Consider the wedge square $Λ^2W$.
	We interpret $w∧W$ as the collection of tensors of the form $∑a_iw∧w_i$,
	that is, the formal sums where the first component is always $w$.
	We interpret $W∧w$ as the collection of tensors of the form $∑a_iw_i∧w$,
	which is the same subset as $w∧W$.
	For higher exterior powers,
	one can interpret $w_1∧w_2∧W$, $w_1∧W∧w_2$, $w_1∧W∧W∧w_2$, etc.\ similarly.
	It is easy to verify that
	they are all subspaces of $Λ^qW$ for the obvious choices of $q$.
	In particular, $W∧W∧\dotsb∧W=Λ^qW$.

讀
\section
					\PT MSR†\MSR/† Product-Matrix in Algebra					

\label{sec:pmalgebra}

	The purpose of this section is to introduce
	the multilinear algebra foundation to the classical constructions
	such that it leads to natural generalizations.
	Usage of multilinear algebra in the context of distributed storage
	dates back to a conference presentation \cite{DL17}.

\subsection{The symmetric translation}\label{sec:t2sym}

	This subsection translates the primitive construction
	in \cref{sec:pmsym} into
	the multilinear algebra language reviewed in the last section.
	Recall $n-1≥d=2(k-1)≥2$.

\subsubsection{File format and \PM M$M$}

	Let $X$ be $𝔽^2$.
	Let $Y$ be $𝔽^{k-1}$.
	Let the file be represented by a linear transformation
	\[*ϕ：X⊗S^2Y⟶𝔽.\tag{file format}\]*
	The file size $M$ is the dimension of $X⊗S^2Y$,
	which is $2·(k-1)k/2$.

\subsubsection{Node configuration and $α$}

	For each $h∈[n]$,
	the $h$th node selects star vectors $x🔑h∈X$ and $y🔑h∈Y$.
	And then the node stores the restriction
	\[*ϕ↾x🔑h⊗y🔑h⊙Y∈(x🔑h⊗y🔑h⊙Y)ˇ.\tag{node content}\]*
	The node size $α$ is the dimension of $x🔑h⊗y🔑h⊙Y$, which is $k-1$.
	The axioms are translated as well.
	
	\begin{axi}\label{axit2sym}
		The selection of the star vectors should conform to
		the following three \MDS/ properties.
		\begin{itemize}
			\mdslabel{x2}	Any two $x🔑h$s span $X$.
				That is to say, $\spa⟨x🔑{h_1},x🔑{h_2}⟩=X$
				for all distinct $h_1,h_2∈[n]$.
			\mdslabel{y2}	Any $k-1$ many $y🔑h$s span $Y$.
				That is, $\spa⟨y🔑{h_1},y🔑{h_2}…y🔑{h_{k-1}}⟩=Y$
				for all distinct $h_1,h_2…h_{k-1}∈[n]$.
			\mdslabel{d2}	Any $d$ many $x🔑h⊗y🔑h$ span $X⊗Y$.
				That is, $\spa⟨x🔑{h_1}⊗y🔑{h_1}…x🔑{h_d}⊗y🔑{h_d}⟩=X⊗Y$
				for all distinct $h_1…h_d∈[n]$.
		\end{itemize}
	\end{axi}
	
	Note that \mdsref{y2} coincides with \mdsref{y}.
	Notice a potential identification $x🔑h≔\bma{1&ξ_h}$.
	Then $x🔑h⊗y🔑h$ corresponds to $\bma{y🔑h&ξ_hy🔑h}$.

\subsubsection{Downloading scheme}

	The downloading scheme boils down to whether
	$k$ restrictions $ϕ↾x🔑h⊗y🔑h⊙Y$ recover the file $ϕ$.
	It is equivalent to this.
	
	\begin{pro}
		Assume \mdsref{x2} and \mdsref{y2},
		then a total of $k$ $x🔑h⊗y🔑h⊙Y$ span $X⊗S^2Y$.
	\end{pro}
	
	\begin{proof}[Sketch]
		Consider the first $k$ many $x🔑h⊗y🔑h⊙Y$s.
		Fix distinct $i,j∈[k]$.
		Then $x🔑i⊗y🔑i⊙Y$ contains $x🔑i⊗y🔑i⊙y🔑j$
		and $x🔑j⊗y🔑j⊙Y$ contains  $x🔑j⊗y🔑i⊙y🔑j$.
		So in the span of them is $\spa⟨x🔑i,x🔑j⟩⊗y🔑i⊙y🔑j$.
		By \mdsref{x2}, the latter is $X⊗y🔑i⊙y🔑j$.
		Vary $j$ over $[k]、\{i\}$,
		then they span $X⊗y🔑i⊙Y$ by \mdsref{y2}.
		Vary $i$ over $[k]$, they span $X⊗Y⊙Y$.
	\end{proof}

\subsubsection{Repairing scheme and $β$}

	When the $f$th node fails, the $h$th node, for each $h∈ℋ$, sends
	\[*ϕ(x🔑h⊗y🔑h⊙y🔑f)∈𝔽\tag{help message}\]*
	to the failing node.
	It is an evaluation so $β=1$.
	Whether or not the failing recovers from the help messages
	reduces to whether $ϕ(x🔑h⊗y🔑h⊙y🔑f)$, a total of $d$ of them,
	determine $ϕ↾x🔑f⊗y🔑f⊙Y$.
	An equivalent statement is here.
	
	\begin{pro}
		Assume \mdsref{d2},
		then a total of $d$ $x🔑h⊗y🔑h⊙y🔑f$ span
		$x🔑f⊗y🔑f⊙Y$ for any $f$.
	\end{pro}
	
	\begin{proof}[Sketch]
		\mdsref{d2} reads $d$ $x🔑h⊗y🔑h$ span $X⊗Y$.
		So $d$ $x🔑h⊗y🔑h⊙y🔑f$ span $X⊗Y⊙y🔑f$.
		The latter contains $x🔑f⊗y🔑f⊙Y$.
	\end{proof}

\subsection{The skew translation}\label{sec:t2ext}

	This subsection translates the skew construction
	in \cref{sec:pmext} to the multilinear algebra language.

\subsubsection{File format and \PM M$M$}

	Let $X$ be $𝔽^2$.
	Let $W$ be $𝔽^k$.
	Let the file $ϕ$ be a linear transformation
	\[*ϕ：X⊗Λ^2W⟶𝔽.\tag{file format}\]*
	The file size $M$ is $\dim(X⊗Λ^2W)=k(k-1)$.

\subsubsection{Node configuration and $α$}

	For each $h∈[n]$,
	the $h$th node selects star vectors $x🔑h∈X$ and $w🔑h∈W$.
	And then the node stores
	\[*ϕ↾x🔑h⊗w🔑h∧W∈(x🔑h⊗w🔑h∧W)ˇ.\tag{node content}\]*
	The node size $α$ is thus $k-1$, because $w🔑h∧w🔑h$ vanishes.
	We do not forget translating the axiom.
	
	\begin{axi}
		The selection of the star vectors should fulfill
		the following three \MDS/ properties.
		\begin{itemize}
			\mdslabel[']{x2}	Any two $x🔑h$s span $X$.
			\mdslabel{w2}	Any $k$ many $w🔑h$s span $W$.
				That is, $\spa⟨w🔑{h_1},w🔑{h_2}…w🔑{h_k}⟩=W$
				for all distinct $h_1,h_2…h_k∈[n]$.
			\mdslabel{q2}	Any $d$ many $x🔑h⊗w🔑h$ span $X⊗(W/⟨w🔑f⟩)$
				for every other index $f$.
				That is,
				$\spa⟨x🔑{h_1}⊗w🔑{h_1}…x🔑{h_d}⊗w🔑{h_d}⟩+X⊗w🔑f=X⊗W$
				for all distinct $f,h_1…h_d∈[n]$.
		\end{itemize}
	\end{axi}
	
	Note that \mdsref[']{x2} coincides with the one in \cref{axit2sym},
	and \mdsref{w2} coincides with \mdsref{w}.
	Notice the potential identification $x🔑h≔\bma{1&ξ_h}$.
	Then $x🔑h⊗w🔑h$ corresponds to $\bma{w🔑h&ξ_hw🔑h}$,
	and $X⊗w🔑f$ to $\spa〈\bma{w🔑f&0},\bma{0&w🔑f}〉$.

\subsubsection{Downloading scheme}

	The downloading scheme can be summarized by the following proposition,
	proof of which is omitted for now.
	But it is a special case of the general theorem.
	
	\begin{pro}
		With \mdsref[']{x2} and \mdsref{w2} assumed,
		$k$ $x🔑h⊗w🔑h∧W$ span $X⊗Λ^2W$.
	\end{pro}

\subsubsection{Repairing scheme and $β$}

	The repairing scheme can be summarized by the following proposition,
	proof of which is omitted for now.
	It is a special case of the general theorem.
	
	\begin{pro}
		With \mdsref{q2} assumed, $d$ $x🔑h⊗w🔑h∧w🔑f$
		span $X⊗W∧w🔑f$ for any other $f$.
	\end{pro}

\subsection{Relation to the polynomial construction}

	The key is to identify $Y≔𝔽^{k-1}$ with $𝔽[y]_{k-2}$.
	Then replace $S^2Y$ with $S^2(𝔽[y]_{k-2})$.
	One also identifies $X≔𝔽^2$ with $𝔽⊕𝔽y^{k-1}$.
	Then $X⊗Y≅𝔽[y]_{2k-3}$ as vector spaces.
	It is possible as well to translate
	the skew construction into polynomials.
	Identify $W≔𝔽^k$ with $𝔽[w]_{k-1}$.

\section{Bridge to High-Rate Codes}\label{sec:highrate}

	Recall the product-matrix mechanism provides \MSR/ codes
	with $d=2(k-1)$ and \cref{lem:shortclose} enables $d>2(k-1)$.
	It remains open whether there are \emph{high-rate} codes
	(meaning $d<2(k-1)$ in this context) that share
	a similar, if not the same, design.
	
	The subsequent two subsections mean to motivate a universal construction
	by giving an explicit $(9,5,6,6)$-\MSR/ code
	and then its moderate generalization to all $d=(3/2)(k-1)$ cases.

\subsection{Explicit \PM(9, 5, 6, 6)$(9,5,6,6)$-\PT MSR†\MSR/† code}
	\label{sec:explicit}

	Here is an explicit, ready-to-use $(9,5,6,6)$-\MSR/ code.

\subsubsection{File format and \PM M$M$}

	Let $𝔽$ be of order $16$;
	it could be realized by the quotient ring $𝔽_2[z]/(z^4+z+1)$.
	Let $X$ be $𝔽^3$.
	Let $Y$ be $𝔽^3$.
	(They play distinct roles and should not be identified.)
	Let $S^3Y$ be the symmetric cube.
	Let the file $ϕ$ be any linear transformation
	\[*ϕ：X⊗S^3Y⟶𝔽.\tag{file format}\]*
	The file size $M$ is the dimension of $X⊗S^3Y$.
	Here $\dim(X)=3$ and $\dim(S^3Y)=\bi 53=10$, so $M=30$.

\subsubsection{Node configuration and $α$}

	Let $a_1,a_2…a_9$ be $0$, $z^3$, $z^6$, $z^{-3}$, $z^{-6}$,
	$z^{-1}$, $z^{-2}$, $z^{-4}$, $z^{-8}$, respectively.
	For each $h∈[9]$,
	the $h$th node selects star vectors
	$x🔑h≔\bma{1&a_h^2&a_h^6}∈X$ and $y🔑h≔\bma{1&a_h&a_h^3}∈Y$.
	And then the node stores the restriction
	\[*ϕ↾x🔑h⊗y🔑h⊙S^2Y∈(x🔑h⊗y🔑h⊙S^2Y)ˇ.\tag{node content}\]*
	The node size $α$ is the dimension of $S^2Y$, so $α=6$.

\subsubsection{Downloading scheme}

	\emph{It happens} that any five $x🔑h⊗y🔑h⊙S^2Y$ span $X⊗S^3Y$,
	hence any five node contents recover the file $ϕ$.

\subsubsection{Repairing scheme and $β$}

	Say the $f$th node fails and the first six nodes are commanded to fix it.
	For each $h∈[6]$, the $h$th node sends 
	\[*ϕ↾x🔑h⊗y🔑h⊙Y⊙y🔑f∈(x🔑h⊗y🔑h⊙Y⊙y🔑f)ˇ.\tag{help message}\]*
	The repair bandwidth $β$ is thus $\dim(Y)=3$.
	\emph{It happens} that any six $x🔑h⊗y🔑h⊙Y⊙y🔑f$ span $x🔑f⊗y🔑f⊙S^2Y$,
	hence the repairing scheme works.
	The specification of the $(9,5,6,6)$-\MSR/ code ends here.
	(This is indeed \MSR/ because $β=3=α/(d-k+1)$.)

\subsection{Warm-up \PM(n, k, (3/2)(k-1))$(n,k,(3/2)(k-1))$%
	-\PT MSR†\MSR/† code}\label{sec:warmup}

	In this subsection, we portray a $d=(3/2)(k-1)$ construction
	as a bridge to the general construction in \cref{sec:atrahasis}.
	The first nontrivial $(k,d)$ pair in this vein is $(5,6)$,
	the parameters used in the last subsection.
	The upcomers are $(7,9)$ followed by $(9,12)$ as well as $(11,15)$.
	Claims in this subsection will not be proven
	as their general counterparts in \cref{sec:atrahasis} come with proofs.

\subsubsection{File format and \PM M$M$}

	Let $𝔽$ be a field. 
	Let $X$ still be $𝔽^3$.
	Let $Y$ be $𝔽^{k-2}$.
	(In general, contrary to the previous subsection,
		$\dim(Y)$ need not be equal to $\dim(X)$.)
	Let $S^3Y$ be the symmetric cube.
	Let the file $ϕ$ be any linear transformation
	\[*ϕ：X⊗S^3Y⟶𝔽.\tag{file format}\]*
	The file size $M$ is the dimension of $X⊗S^3Y$.
	Here $\dim(X)=3$ and $\dim(S^3Y)=\bi k3$, so $M=(k-2)(k-1)k/2$.

\subsubsection{Node configuration and $α$}

	For each $h∈[n]$,
	the $h$th node selects star vectors $x🔑h∈X$ and $y🔑h∈Y$.
	And then the node stores the restriction
	\[*ϕ↾x🔑h⊗y🔑h⊙S^2Y∈(x🔑h⊗y🔑h⊙S^2Y)ˇ.\tag{node content}\]*
	The node size $α$ is the dimension of $S^2Y$, so $α=(k-1)(k-2)/2$.
	The selection of the star vectors shall meet three \MDS/ requirements.
	
	\begin{axi}\label{axi:t3sym}
		The selection of the star vectors are such that:
		\begin{itemize}
			\mdslabel{x3}	Any three $x🔑h$s span $X$.
				That is, $\spa⟨x🔑{h_1},x🔑{h_2},x🔑{h_3}⟩=X$
				for all distinct $h_1,h_2,h_3∈[n]$.
			\mdslabel{y3}	Any $k-2$ many $y🔑h$s span $Y$.
				That is, $\spa⟨y🔑{h_1}…y🔑{h_{k-2}}⟩=Y$
				for all distinct $h_1…h_{k-2}∈[n]$.
			\mdslabel{d3}	Any $d$ many $x🔑h⊗y🔑h⊙Y$ span $X⊗S^2Y$.
				That is,
				$x🔑{h_1}⊗y🔑{h_1}⊙Y+\dotsb+x🔑{h_d}⊗y🔑{h_d}⊙Y=X⊗S^2Y$
				for all distinct $h_1…h_d∈[n]$.
		\end{itemize}
	\end{axi}
	
	It is unclear whether there are easy ways to generate star vectors.
	There are some heuristics that suggest hopeful patterns;
	accordingly we found some working instances by brute force.
	See \cref{sec:force}.

\subsubsection{Downloading scheme}

	Say the first $k$ nodes are retrieved.
	For any distinct indices $h,i,j∈[k]$, we extract $ϕ(x🔑h⊗y🔑h⊙y🔑i⊙y🔑j)$,
	$ϕ(x🔑i⊗y🔑h⊙y🔑i⊙y🔑j)$, and $ϕ(x🔑j⊗y🔑h⊙y🔑i⊙y🔑j)$
	from the $h$th, the $i$th, and the $j$th nodes, respectively.
	From there we learn $ϕ↾X⊗y🔑h⊙y🔑i⊙y🔑j$ by \mdsref{x3}.
	Next, we learn $ϕ↾X⊗y🔑h⊙y🔑i⊙Y$ by
	applying \mdsref{y3} to various $j∈[k]、\{h,i\}$.
	Once done, we vary $i$ to study $ϕ↾X⊗y🔑h⊙Y⊙Y$ by \mdsref{y3}.
	The latter then helps us reestablish $ϕ↾X⊗Y⊙Y⊙Y$, which is the file per se.

\subsubsection{Repairing scheme and $β$}

	Say the $f$th node fails and the first $d$ nodes are commanded to fix it.
	For each $h∈[d]$, the $h$th node sends 
	\[*ϕ↾x🔑h⊗y🔑h⊙Y⊙y🔑f∈(x🔑h⊗y🔑h⊙Y⊙y🔑f)ˇ.\tag{help message}\]*
	The repair bandwidth $β$ is thus $\dim(Y)=k-2$.
	By \mdsref{d3}, the failing node learns $ϕ↾X⊗Y⊙Y⊙y🔑f$
	from the help messages.
	And then the node specializes it to $ϕ↾x🔑f⊗Y⊙Y⊙y🔑f$,
	which is $ϕ↾x🔑f⊗y🔑f⊙Y⊙Y$.
	One can see here that the validity of the repairing scheme
	depends entirely on whether \mdsref{d3}, the third \MDS/ axiom, holds.
	The subtlety is how to design star vectors.
	
	We close this section with a remark that codes defined in this subsection
	are special cases of the general code in \cref{sec:tallsym}.
	\Cref{axi:t3sym}, for instance, is a special case of \cref{axi:tallsym}.
	See \cref{tab:accents} for complete relations among all constructions.
	
	\pgfplotstableread[header=false]{
		~				$t$		symmetric		exterior		polynomial		
		product-matrix	$2$		\ref{pmsym}		\ref{pmext}		\ref{pmpoly}	
		\cmidrule(lr){1-5}%
		~				$2$		\ref{t2sym}		\ref{t2ext}		--				
		\MLA/			$3$		\ref{warmup}	--				\ref{t3poly}	
		~				$≥2$	\ref{tallsym}	\ref{tallext}	--				
	}\pgfAccents
	\begin{table}
		\caption{
			Product-matrix constructions (top left cell),
			its accents (first row), and their generalizations (other rows).
			Omitted entries are omitted to reduce repetition,
			not because they are not possible.
			Note that polynomials come in two flavors
			(symmetric and anti-symmetric/skew-symmetric/alternating/exterior).
			But we only talk about the symmetric ones.
		}\label{tab:accents}
		\def\arraystretch{1.44}
		\def\MLA/{\tikz[overlay]
			\node[align=center]{multilinear\strut\\algebra\strut};}
		\def\ref#1{\cref{sec:#1}}
		\centering\pgfplotstabletypeset[
			every head row/.style=output empty row,
			every row no 0/.style={before row=\toprule,after row=\midrule},
			every last row/.style={after row=\bottomrule},
			string type,
		]\pgfAccents
	\end{table}

讀
\section
								 Atrahasis Code 								

\label{sec:atrahasis}

	We specify and verify the general Atrahasis code in this section.
	That will prove \cref{thm:primitive}.
	Recall that the parameters we are interested in are
	integers $n,k,d,t$ such that $n-1≥d≥k≥2$ and $t=d/(d-k+1)$.
	We invite readers to organize parameters in this form.
	
	\begin{pro}
		Asterisks are unimportant place holders.
		The matrix
		\[*\bma{
			d-k+1	&	k-1	&	d	&	α	&	α\log_2\abs{𝔽}	\\
			1		&	t-1	&	t	&	β	&	β\log_2\abs{𝔽}	\\
			*		&	*	&	kd	&	M	&	M\log_2\abs{𝔽}	
		}\]*
		is of rank one.
	\end{pro}
	
	\begin{proof}
		Trivial.
	\end{proof}
	
	For \cref{thm:primitive}, we provide two proofs.
	One utilizes symmetric power and the other leans on exterior power.

\subsection{Symmetric power proof of Theorem~\ref{thm:primitive}}
	\label{sec:tallsym}

\subsubsection{File format and \PM M$M$}

	Let $X$ be $𝔽^t$.
	Let $Y$ be $𝔽^{k-t+1}$.
	Let $S^tY$ be the $t$th symmetric power of $Y$.
	Let the file $ϕ$ be encoded as a linear transformation
	\[*ϕ：X⊗S^tY⟶𝔽.\tag{file format}\]*
	This arbitrary map in $(X⊗S^tY)ˇ$ is able to carry $\dim(X⊗S^tY)$ symbols.
	Here $\dim(X)=t$ and $\dim(S^tY)=\bi kt$.
	Therefore, $M=t\bi kt=k\bi{k-1}{t-1}$.

\subsubsection{Node configuration and $α$}

	Let $[n]≔\{1,2…n\}$ represent the set of nodes.
	For each $h∈[n]$, the $h$th node selects two star vectors:
	$x🔑h∈X$ and $y🔑h∈Y$.
	Next, the $h$th node stores a restriction of the file
	\[*ϕ↾x🔑h⊗y🔑h⊙S^{t-1}Y∈(x🔑h⊗y🔑h⊙S^{t-1}Y)ˇ.\tag{node content}\]*
	This restriction is a linear transformation
	from the domain $x🔑h⊗y🔑h⊙S^{t-1}Y$.
	As a consequence,
	it can be fully recorded by $\dim(x🔑h⊗y🔑h⊙S^{t-1}Y)$ symbols.
	This quantity coincides with $\dim(S^{t-1}Y)$, which is $\bi{k-1}{t-1}$.
	In summary, the \subpack/ level is $α=\bi{k-1}{t-1}$.
	
	For the downloading scheme and repairing scheme later in the proof,
	the selection of star vectors $x🔑h$s and $y🔑h$s are not arbitrary.
	There are several conditions they need to fulfill.
	
	\begin{axi}\label{axi:tallsym}
		Assume that the selection of the star vectors
		fulfills the following three \MDS/ conditions.
		\begin{itemize}
			\mdslabel{xt}	Any $t$ many $x🔑h$s span $X$.
				Namely, $\spa⟨x🔑{h_1}…x🔑{h_t}⟩=X$
				for all distinct indices $h_1…h_t∈[n]$.
			\mdslabel{yt}	Any $k-t+1$ many $y🔑h$s span $Y$.
				To wit, $\spa⟨y🔑{h_1}…y🔑{h_{k-t+1}}⟩=Y$
				for all distinct indices $h_1…h_{k-t+1}∈[n]$.
			\mdslabel{dt}	Any $d$ many $x🔑h⊗y🔑h⊙S^{t-2}Y$
				span $X⊗S^{t-1}Y$.
				More specifically,
				$x🔑{h_1}⊗y🔑{h_1}⊙S^{t-2}Y+\dotsb+x🔑{h_d}⊗y🔑{h_d}⊙S^{t-2}Y$
				is $X⊗S^{t-1}Y$ for all distinct indices $h_1…h_d∈[n]$.
		\end{itemize}
	\end{axi}
	
	This axiom generalizes \cref{axi:t3sym}.
	As commented there, it is unclear how such star vectors can be found easily.
	The existence of star vectors, on the other hand,
	is guaranteed by Alon's combinatorial Nullstellensatz.
	That being said, we have no control over
	the upper bound on field size other than Alon's.
	(Bounds from DeMillo--Lipton--Schwartz--Zippel,
		if not coincident, are looser.)
	See \cref{sec:alon} for more details.

\subsubsection{Downloading scheme}

	To verify that downloading any $k$ nodes
	suffices to recover the whole file $ϕ$,
	let $𝒦⊆[n]$ be the indices of downloaded nodes, $\abs{𝒦}=k$.
	We now possess complete knowledge
	of $ϕ↾x🔑h⊗y🔑h⊙S^{t-1}Y$ for all $h∈𝒦$.
	Provided that $ϕ$ is linear, we recover the restriction to the span
	\[*ϕ
		\mathbin{
			\rlap{\raisebox{-3pt}{$↾$}}
			\rlap{\tikz\fill[white,opacity=1]rectangle(.2,.1)(,-.05);}
			\raisebox{3pt}{$↾$}
		}
		∑_{h∈𝒦}x🔑h⊗y🔑h⊙S^{t-1}Y.\]*
	Whether or not this is $ϕ$ per se depends on
	whether or not the span is the original domain, $X⊗S^tY$.
	
	\begin{pro}\label{pro:tallsymdown}
		Take \mdsref{xt} and \mdsref{yt} as granted.
		For any $k$-element subset $𝒦⊆[n]$,
		\[*∑_{h∈𝒦}x🔑h⊗y🔑h⊙S^{t-1}Y=X⊗S^tY.\]*
	\end{pro}
	
	\begin{proof}
		Let $i_1,i_2…i_t∈𝒦$ be $t$ distinct indices.
		Then $∑_{h∈𝒦}x🔑h⊗y🔑h⊙S^{t-1}Y$ contains the following $t$ tensors
		\begin{align*}
				x🔑{i_1}⊗y🔑{i_1}⊙y🔑{i_2}⊙y🔑{i_3}⊙\dotsb⊙y🔑{i_t},\\
			x🔑{i_2}⊗y🔑{i_2}⊙y🔑{i_1}⊙y🔑{i_3}⊙\dotsb⊙y🔑{i_t}
				=x🔑{i_2}⊗y🔑{i_1}⊙y🔑{i_2}⊙y🔑{i_3}⊙\dotsb⊙y🔑{i_t},\\
			x🔑{i_3}⊗y🔑{i_3}⊙y🔑{i_1}⊙y🔑{i_2}⊙\dotsb⊙y🔑{i_t}
				=x🔑{i_3}⊗y🔑{i_1}⊙y🔑{i_2}⊙y🔑{i_3}⊙\dotsb⊙y🔑{i_t},\\
					⋮\kern3em\\
			x🔑{i_j}⊗y🔑{i_j}⊙y🔑{i_1}⊙y🔑{i_2}
			⊙\dotsb\widehat{y🔑{i_j}}\dotsb⊙y🔑{i_t}
				=x🔑{i_j}⊗y🔑{i_1}⊙y🔑{i_2}⊙y🔑{i_3}⊙\dotsb⊙y🔑{i_t},\\
					⋮\kern3 em\\
			x🔑{i_t}⊗y🔑{i_t}⊙y🔑{i_1}⊙y🔑{i_3}⊙\dotsb⊙y🔑{i_{t-1}}
				=x🔑{i_t}⊗y🔑{i_1}⊙y🔑{i_2}⊙y🔑{i_3}⊙\dotsb⊙y🔑{i_t}.
		\end{align*}
		A vector under a wide hat is missing from the product.
		Note that the right column consists of ${}⊗y🔑{i_1}⊙\dotsb⊙y🔑{i_t}$
		led by $x🔑{i_1}$, $x🔑{i_2}$, and all the way up to $x🔑{i_t}$.
		By the distributive law, tensors in the right column span
		\[*\spa⟨x🔑{i_1},x🔑{i_2}…x🔑{i_t}⟩
			⊗y🔑{i_1}⊙y🔑{i_2}⊙y🔑{i_3}⊙\dotsb⊙y🔑{i_t}.\]*
		Invoking \mdsref{xt}, this subspace is
		\[*X⊗y🔑{i_1}⊙y🔑{i_2}⊙y🔑{i_3}⊙\dotsb⊙y🔑{i_t}.\]*
		Let $i_t$ vary over $𝒦、\{i_1…i_{t-1}\}$
		(all possible indices such that all subscripts are distinct).
		Then these subspaces sum to
		\[*X⊗y🔑{i_1}⊙\dotsb⊙y🔑{i_{t-1}}
			⊙\spa⟨y🔑{i_t}:i_t∈𝒦、\{i_1…i_{t-1}\}⟩.\]*
		According to \mdsref{yt}, the span can be replaced by $Y$.
		Thus $∑_{h∈𝒦}x🔑h⊗y🔑h⊙S^{t-1}Y$ contains
		\[*X⊗y🔑{i_1}⊙\dotsb⊙y🔑{i_{t-1}}⊙Y\]*
		for any $i_1…i_{t-1}∈𝒦$.
		Now we replicate the same procedure to replace $y_{i_{t-1}}$ by $Y$,
		and then replace $y_{i_{t-2}}$ by $Y$.
		In the end, we show that $∑_{h∈𝒦}x🔑h⊗y🔑h⊙S^{t-1}Y$ contains
		\[*X⊗Y⊙\dotsb⊙Y,\]*
		which is the domain of $ϕ$.
	\end{proof}

\subsubsection{Repairing scheme and $β$}

	Let $f∈[n]$ be the index that points to the failing node.
	Let $ℋ⊆[n]、\{f\}$ be the $d$ indices,
	$\abs{ℋ}=d$, that point to the helper nodes.
	When the $f$th node fails, each helper node $h∈ℋ$ sends the restriction
	\[*ϕ↾x🔑h⊗y🔑h⊙S^{t-2}Y⊙y🔑f∈(x🔑h⊗y🔑h⊙S^{t-2}Y⊙y🔑f)ˇ
		\tag{help message}\]*
	to the former.
	The $h$th node knows what to send because the help message is
	a further restriction (to a smaller subspace) of its node content.
	In particular, $x🔑h⊗y🔑h⊙S^{t-2}Y⊙y🔑f⊆x🔑h⊗y🔑h⊙S^{t-1}Y$.
	In sending the help message, the helper node needs to transmit
	$\dim(x🔑h⊗y🔑h⊙S^{t-2}Y⊙y🔑f)$ symbols.
	This is $\dim(S^{t-2}Y)$, or $\bi{k-2}{t-2}$ for short.
	So the repair bandwidth is $β=\bi{k-2}{t-2}$.
	Now the help messages are sent.
	
	Upon the reception of help messages,
	the failing node recalls its original content
	if the corresponding subspaces span its domain.
	More precisely, it relies on the following containment.
	
	\begin{pro}
		Take \mdsref{dt} as granted.
		For any $d$-subset $ℋ⊆[n]、\{f\}$,
		\[*∑_{h∈ℋ}x🔑h⊗y🔑h⊙S^{t-2}Y⊙y🔑f⊇x🔑f⊗y🔑f⊙S^{t-1}Y.\]*
	\end{pro}
	
	\begin{proof}
		Specialize \mdsref{dt} at $ℋ$.
		We obtain
		\[*∑_{h∈ℋ}x🔑h⊗y🔑h⊙S^{t-2}Y=X⊗S^{t-1}Y.\]*
		Citing the distributive law, we further deduce that
		\[*∑_{h∈ℋ}x🔑h⊗y🔑h⊙S^{t-2}Y⊙y🔑f=X⊗S^{t-1}Y⊙y🔑f.\]*
		The right-hand side is $X⊗y🔑f⊙S^{t-1}Y$;
		the latter clearly contains a subspace $x🔑f⊗y🔑f⊙S^{t-1}Y$.
		And we are done proving.
	\end{proof}
	
	\Cref{thm:primitive}'s proof is now complete up to \cref{axi:tallsym}
	(which is closely related to the field size part of the theorem statement).
	We defer that part until \cref{sec:alon}.
	One also sees that this subsection specializes
	to \cref{sec:t2sym} when $t=2$, and to \cref{sec:warmup} when $t=3$.
	The rest of this section is an alternative proof
	of \cref{thm:primitive} utilizing exterior power.

\subsection{Exterior power proof of Theorem~\ref{thm:primitive}}
	\label{sec:tallext}

\subsubsection{File format and \PM M$M$}

	Let $X$ be $𝔽^t$.
	Let $W$ be $𝔽^k$.
	Let $Λ^tW$ be the $t$th exterior power.
	Let the file $ϕ$ be any linear transformation
	\[*ϕ：X⊗Λ^tW⟶𝔽.\tag{file format}\]*
	The file size is thus the dimension of $X⊗Λ^tW$,
	which means $M=t\bi kt=k\bi{k-1}{t-1}$.
	
\subsubsection{Node configuration and $α$}

	For each $h∈[n]$,
	the $h$th node selects star vectors $x🔑h∈X$ and $w🔑h∈W$.
	And then the node stores the restriction
	\[*ϕ↾x🔑h⊗w🔑h∧Λ^{t-1}W.\tag{node content}\]*
	The node size $α$ is the dimension of this subspace,
	which is $\dim(w🔑h∧Λ^{t-1}W)$.
	Notice that we do not automatically equal it to $\dim(Λ^{t-1}W)$.
	This is because $w🔑h∧{}$ will eliminate a tensor
	whenever its $Λ^{t-1}W$-fragment is a multiple of $w🔑h$,
	viz.\ $w🔑h∧w🔑h∧ω=0$ for all $ω∈Λ^{t-2}W$.
	To rephrase it, the $Λ^{t-1}W$-fragment contributes,
	and only contributes, tensors ``up to $w🔑h$''.
	
	\begin{lem}\label{lem:quot}
		Let $w∈W$, then $w∧Λ^{t-1}W≅Λ^{t-1}(W/⟨w⟩)$ as a vector space.
	\end{lem}
	
	\begin{proof}
		We claim the desired linear isomorphism
		\[*w∧w_1∧w_2∧\dotsb∧w_{t-1}
			⟼(w_1+⟨w⟩)∧(w_2+⟨w⟩)∧\dotsb∧(w_{t-1}+⟨w⟩).\]*
		One can confirm that this map is well-defined,
		linear, injective, and surjective.
	\end{proof}
	
	With the lemma, we argue that $α=\dim(w🔑h∧Λ^{t-1}W)
		=\dim(Λ^{t-1}(W/⟨w🔑h⟩))=\dim(Λ^{t-1}𝔽^{k-1})$.
	So the node size is indeed $α=\bi{k-1}{t-1}$.
	Next, we state the axioms concerning the star vectors.
	
	\begin{axi}\label{axi:tallext}
		The selection of the star vectors satisfies
		the following three \MDS/ conditions.
		\begin{itemize}
			\mdslabel[']{xt}	Any $t$ many $x🔑h$s span $X$.
				That is, $\spa⟨x🔑{h_1}…x🔑{h_t}⟩=X$
				for all distinct $h_1…h_t∈[n]$.
			\mdslabel{wt}	Any $k$ many $w🔑h$s span $W$.
				That is, $\spa⟨w🔑{h_1}…w🔑{h_k}⟩=W$
				for all distinct $h_1…h_k∈[n]$.
			\mdslabel{qt}	Any $d$ many $x🔑h⊗w🔑h∧Λ^{t-2}W$ span
				$X⊗Λ^{t-1}(W/⟨w🔑f⟩)$ for every other index~$f$.
				That is to say, $x🔑{h_1}⊗w🔑{h_1}∧Λ^{t-2}W+\dotsb+
					x🔑{h_d}⊗w🔑{h_d}∧Λ^{t-2}W+X⊗w🔑f∧Λ^{t-2}W$
				is $X⊗Λ^{t-1}W$ for all distinct $f,h_1…h_d∈[n]$.
		\end{itemize}
	\end{axi}
	
	\mdsref[']{xt} coincides with the one in \cref{axi:tallsym}.
	\Cref{axi:tallext} is a generalization of \cref{axi:pmext}.
	Remarks under \cref{axi:t3sym,axi:tallsym}
	(that we do not have efficient algorithm to generate star vectors)
	also apply here.
	See \cref{sec:field} for how we overcome this.

\subsubsection{Downloading scheme}

	Whether or not any $k$ node contents recover the file $ϕ$
	is equivalent to whether any $k$ corresponding domains span $ϕ$'s.
	We end up relying on this proposition.
	
	\begin{pro}
		Assume \mdsref[']{xt} and \mdsref{wt}.
		Let $𝒦⊆[n]$ be a $k$-subset.
		Then
		\[*∑_{h∈𝒦}x🔑h⊗w🔑h∧Λ^{t-1}W=X⊗Λ^tW.\]*
	\end{pro}
	
	\begin{proof}[Sketch]
		Similar strategy to \cref{pro:tallsymdown}.
		First, obtain $X⊗w🔑{i_1}∧w🔑{i_2}∧\dotsb∧w🔑{i_t}$.
		And then replace lower letter $w$'s by capital $W$, one after another.
		In doing so, use the free knowledge $w∧w=0$.
	\end{proof}

\subsubsection{Repairing scheme and $β$}

	The $h$th node, for each helper index, sends to the $f$th node,
	the failing node, the restriction
	\[*ϕ↾x🔑h⊗w🔑h∧Λ^{t-2}W∧w🔑f.\tag{help message}\]*
	Subspace $x🔑h⊗w🔑h∧Λ^{t-2}W∧w🔑f$ is contained
	in the node domain $x🔑h⊗w🔑h∧Λ^{t-1}W$.
	Subspace $x🔑h⊗w🔑h∧Λ^{t-2}W∧w🔑f$ has dimension
	$\dim(w🔑h∧Λ^{t-2}W∧w🔑f)$.
	Invoking \cref{lem:quot}, twice, we can write $w🔑h∧Λ^{t-2}W∧w🔑f
		≅Λ^{t-2}(W/\spa⟨w🔑h⟩)∧w🔑f≅Λ^{t-2}(W/\spa⟨w🔑h,w🔑f⟩)$.
	Hence the dimension is $\dim(W/\spa⟨w🔑h,w🔑f⟩)$ choose $t-2$,
	that will lead to $β=\bi{k-2}{t-2}$.
	The effectiveness of repairing is handled below.
	
	\begin{pro}
		Assume \mdsref{qt}.
		Let $f∈[n]$ and let $ℋ⊆[n]、\{f\}$ be such that $\abs{ℋ}=d$.
		Then
		\[*∑_{h∈ℋ}x🔑h⊗w🔑h∧Λ^{t-2}W∧w🔑f⊇x🔑f⊗w🔑f∧Λ^{t-1}W.\]*
	\end{pro}
	
	\begin{proof}
		Multiply \mdsref{qt} by ${}∧w🔑f$ from the right.
		Replace $X$ by $x🔑f$.
	\end{proof}
	
	This finishes the proof of \cref{thm:primitive}
	modulo field size for the second time.
	One also sees that this subsection specializes
	to \cref{sec:t2ext} when $t=2$.
	In the next section, we deal with the elephant in the room.

讀
\section
						 Star Selection and Field Size  						

\label{sec:field}

	We left open how nodes select star vectors such that
	\mdsref{x3}, \mdsref{y3}, and \mdsref{d3} in \cref{sec:warmup} hold.
	And then in \cref{sec:tallsym} we assume
	\mdsref{xt}, \mdsref{yt}, and \mdsref{dt} without specifying how.
	Nor did we disclose how to fulfill
	\mdsref[']{xt}, \mdsref{wt}, and \mdsref{qt} in \cref{sec:tallext}.
	In this section, we propose two approaches.
	One is an existence bound (as commented below \cref{axi:tallsym}).
	The other is by brute force.

\subsection{A loose bound}\label{sec:alon}

	Recall N.~Alon's combinatorial Nullstellensatz.
	
	\begin{lem}\label{lem:alon}
		\cite[Theorem~1.2]{Alon99}
		Let $𝔽$ be a field.
		Let $t_1…t_n$ be nonnegative integers.
		Let $f(x_1…x_n)$ be a polynomial over $𝔽$ in $n$ variables.
		Suppose $\deg f=t_1+\dotsb+t_n$ and
		the coefficient of $x_1^{t_1}\dotsm x_n^{t_n}$ in $f$ is nonzero.
		Let $S_1…S_n⊆𝔽$ be any subsets with $\abs{S_i}>t_i$ for all $i∈[n]$.
		Then $f(s_1…s_n)≠0$ for some $s_1∈S_1$, $s_2∈S_2$,
		and all the way up to $s_n∈S_n$.
	\end{lem}
	
	A common use of the combinatorial Nullstellensatz is to
	insert variables into a square matrix that is presumed to be invertible.
	Imagine its determinant being a multivariate polynomial.
	If this polynomial is nonzero,
	one can find a top total-degree monomial within.
	Its degree will be the $t_1…t_n$ in the statement and
	the lower bounds on the sizes of $S_1…S_n$.
	Subsets $S_1…S_n$ are usually assumed to be the field $𝔽$ itself
	so $t_1…t_n$ serve as lower bounds on the field size.
	
	Frequently it is the case that all we need is a finite bound,
	so we do not have to keep track of $t$'s.
	In such circumstances, Alon's theorem reads:
	A nontrivial polynomial has a nonzero evaluation.
	Notice its elementary converse---%
	nonzero evaluation implies nonzero polynomial.
	
	Now what we demand is the existence of star vectors $x🔑1,y🔑1…x🔑n,y🔑n$
	that satisfy \mdsref{xt}, \mdsref{yt}, and \mdsref{dt}.
	Take \mdsref{xt} as an example.
	Whether or not any $t$ vectors among $x🔑1…x🔑n∈X$ span $X$
	is equivalent to whether any $t$ vectors
	form a $t$-by-$t$ matrix with a nonzero determinant.
	Let $f(x🔑1…x🔑t)$ be the determinant
	written as a polynomial in the coordinates of $x🔑1…x🔑t$.
	Then we want to show
	\[*∏_{i_1…i_t}f(x🔑{i_1}…x🔑{i_t})≠0\]*
	as a polynomial.
	This is true because plugging in Reed--Solomon columns
	results in a nonzero evaluation.
	Similarly, for \mdsref{yt},
	let $g(y🔑1…y🔑{k-t+1})$ be the determinant in terms of $y$'s.
	Then
	\[*∏_{i_1…i_{k-t+1}}g(y🔑{i_1}…y🔑{i_{k-t+1}})\]*
	is, again, not a zero polynomial due to Reed--Solomon codes.
	
	Up to this point, it remains to show that
	\[*∏_{i_1…i_d}h(x🔑{i_1},y🔑{i_1}…x🔑{i_d},y🔑{i_d})\]*
	is nonzero, where $h$ is the determinant corresponding to \mdsref{dt}.
	This one is hard, because we do not know any code that
	guarantees \emph{not} to evaluate $h$ to zero.
	Nonetheless, there is a shenanigan to overcome small cases.
	
	\begin{alg}\label{alg:rand}
		We executed the following for all $α≤3003$ cases.
		\begin{enumerate}
			\item Let $𝔽'$ be a finite field of small prime order.
				For instance $\abs{𝔽'}=127$.
			\item Let $x🔑1…x🔑d∈𝔽'^t$ and $y🔑1…y🔑d∈𝔽'^{k-t+1}$
				be random vectors of the prescribed lengths
				drawn from any ensemble.
			\item Select a basis $¯{η}_1…¯{η}_{\bi{k-2}{t-2}}$
				of $S^{t-2}𝔽'^{k-t+1}$.
				Select for $S^{t-1}𝔽'^{k-t+1}$, too.
				(The standard ones in \cref{sec:algebra} are preferred.)
			\item Expand $x🔑i⊗y🔑i⊙¯{η}_j$
				for all $i∈[d]$ and all $j∈［\bi{k-2}{t-2}］$
				as very long vectors in $𝔽'^{dβ}$,
				and stack them to form a $dβ$-by-$dβ$ matrix.
			\item Compute the determinant $h(x🔑1,y🔑1…x🔑d,y🔑d)∈𝔽'$
				of the matrix.
				If it is nonzero,
				then $h$ has a nonzero evaluation and hence is nonzero.
				We declare a pass.
				Otherwise redraw random vectors and start over.
		\end{enumerate}
	\end{alg}
	
	Remarks:
	All $α≤3003$ cases passed;
	some did require a second run as the determinant vanished in the first run.
	Computing over a finite field $𝔽'$ in place of $ℚ$ (or floating numbers)
	is essential because the arithmetic is exact and fast.
	The sole purpose of the field $𝔽'$ is to witness $h≠0$ over $ℤ$,
	so it does not have to be the same field $𝔽$ we define the actual code over.
	A smaller field causes a faster computation with a lower pass rate.
	The result of \cref{alg:rand} can be summarized as follows.
	
	\begin{pro}\label{pro:nonzero}
		For all $α≤3003$ cases, the determinant
		$h(x🔑1,y🔑1…x🔑d,y🔑d)$ is not the zero polynomial over $ℤ$.
	\end{pro}
	
	\Cref{pro:nonzero,lem:alon} jointly imply \cref{axi:tallsym}
	for all $α≤3003$ cases, which completes the proof of
	\cref{thm:primitive} on the basis of \cref{sec:tallsym}.
	And we are done proving our main theorem
	if readers are satisfied with $α≤3003$.
	Otherwise, here is a conditional result.
	
	\begin{pro}
		If, for some $k,d,t$, the determinant $h(x🔑1,y🔑1…x🔑d,y🔑d)$ is
		not the zero polynomial in the coordinates of $x🔑1,y🔑1…x🔑d,y🔑d$,
		then an $(n,k,d,α)$-\MSR/ code exists over some sufficiently large field.
		In particular, if $h$ is never a zero polynomial,
		then \cref{con:all} holds.
	\end{pro}
	
	Remark:
	Sometimes, in place of the combinatorial Nullstellensatz,
	the DeMillo--Lipton--Schwartz--Zippel lemma is cited.
	The lemma reads:
	Let $S$ be a finite subset of a field $𝔽$.
	Let $f(x_1…x_n)$ be a degree-$t$ polynomial over $𝔽$.
	Select $s_1…s_n∈S$ independently, uniformly at random.
	Then $f(s_1…s_n)=0$ with probability at most $t/\abs S$.
	This lemma gives a strictly worse bound on the field size
	since $t$ is the ``$l^1$-degree'',
	while the Nullstellensatz deals with the ``$l^∞$-degree''.

\subsection{Brute force}\label{sec:force}

	\pgfplotstableread[header=false]{
		t	\abs{𝔽}	n	k	d		α	β	M		
		2	O(n)	>d	k	2(k-1)	k-1	1	k(k-1)	
		3	16		9	5	6		6	3	30		
		·	256		13	·	·		·	·	·		
		·	2048	17	·	·		·	·	·		
		·	32		11	7	9		15	5	105		
		·	256		13	·	·		·	·	·		
		·	1024	15	·	·		·	·	·		
		·	32		16	9	12		28	7	252		
		·	64		16	11	15		45	9	495		
		·	64		19	13	18		66	11	858		
		4	32		10	7	8		20	10	140		
		·	256		12	·	·		·	·	·		
		·	64		13	10	12		84	28	840		
		∫·	∫128	∫14	∫·	∫·		∫·	∫·	∫·		^^M 
		·	512		15	·	·		·	·	·		
		5	128		11	9	10		70	35	630		
		k	2		k+1	k	k		1	1	k		
		·	n		>k	·	·		·	·	·		
	}\pgfSageforce
	\begin{table}
		\caption{
			Parameter tuples with known instances found by brute force
			(jointly with some clever heuristics). 
			Omitted entries inherit values from upper neighbors.
			One sees that the field size grows exponentially in the node number.
		}\label{tab:sageforce}
		\pgfplotstablegetrowsof\pgfSageforce
		\PMT\lastrow{\pgfplotsretval-1}
		\pgfplotsforeachungrouped\r in{2,...,\lastrow}{
			\def\textcdot{·}
			\pgfplotstablegetelem{\numexpr\r-1}3\of\pgfSageforce
			\ifx\textcdot\pgfplotsretval
			\pgfplotstablegetelem\r3\of\pgfSageforce
			\ifx\textcdot\pgfplotsretval\else
				\pgfplotstableset{every row no \r/.style={before row=\krule}}
			\fi\fi
			\pgfplotstablegetelem\r0\of\pgfSageforce
			\ifx\textcdot\pgfplotsretval\else
				\pgfplotstableset{every row no \r/.style={before row=\trule}}
			\fi
			\def\textintcdot{∫·}\def\int{\color{white}}
			\ifx\textintcdot\pgfplotsretval
				\pgfplotstableset{
					every row no \r/.style={before row=\rowcolor{Periwinkle}}}
			\fi
		}
		\def\trule{\cmidrule(lr){1-8}}	\def\krule{\cmidrule(lr{1.5em}){2-8}}
		\centering\pgfplotstabletypeset[
			every head row/.style=output empty row,
			every row no 0/.style={before row=\toprule,after row=\midrule},
			every last row/.style={after row=\bottomrule},
			assign cell content/.style={@cell content/.initial=$#1$},
		]\pgfSageforce
	\end{table}
	
	Throughout \cref{axi:pmsym,axi:pmext,axi:t3sym,axi:tallsym,axi:tallext},
	the first two conditions are always easy to fulfill.
	One queries the list of $[n,t]$-\MDS/ codes over a chosen field
	and let $x🔑h$ be the columns of the generator matrix of a chosen code.
	Similarly, one chooses an $[n,k]$ (or $[n,k-t+1]$)-\MDS/ code
	and let $w🔑h$ (or $y🔑h$) be the columns of its generator matrix.
	However, that does not say anything about
	whether the third, be it \mdsref{dt} or \mdsref{qt}, is met.
	To demonstrate our strategy for generating practical codes,
	pretend that we want to build an $(n,5,6,6)$-\MSR/ code.
	
	\begin{alg}
		Here is what we did.
		\begin{enumerate}
			\item	Let $𝔽$ be of order $16$;
				it could be realized by $𝔽_{16}≔𝔽_2[z]/(z^4+z+1)$.
			\item	Each node chooses a unique point $a_h∈𝔽$.
			\item	The $x$-vectors are of the form
				$x🔑h≔\bma{1&a_h^2&a_h^6}$.
			\item	The $y$-vectors are of the form
				$y🔑h≔\bma{1&a_h&a_h^3}$.
			\item	Enumerate and assert \mdsref{xt},
				\mdsref{yt}, and \mdsref{dt}.
		\end{enumerate}
	\end{alg}
	
	The largest pool of points we can find is
	$\{0,z^3,z^6,z^{-3},z^{-6},z^{-1},z^{-2},z^{-4},z^{-8}\}$,
	as is posed in \cref{sec:explicit}.
	Therefore, we announce that
	there exists an $(9,5,6,6)$-\MSR/ code over $𝔽_{16}$.
	On the basis of this example, variables of this ensemble are as follows.
	\begin{itemize}
		\item	$n$ the number of nodes depends on how many points
			can be added to the point pool before
			\mdsref{xt}, \mdsref{yt}, or \mdsref{dt} breaks.
		\item	One should try a larger field
			in order to find a larger collection of points.
		\item	One can try a different pattern for $x$-vectors,
			for instance $[1\ a_h\ a_h^2]^⊤$ and $[1\ a_h^3\ a_h^4]^⊤$.
		\item	One may try a different pattern for $y$-vectors.
		\item	One may try to fulfill \mdsref[']{xt}, \mdsref{wt},
			and \mdsref{qt} instead, i.e., the skew version.
	\end{itemize}
	It is unclear at this stage
	what is the best practice to find the point pool.
	So far brute force works better than any heuristics alone.
	We devoted some computing resources and
	the results are listed in \cref{tab:sageforce}.
	

\subsection{General polynomial shorthand}\label{sec:t3poly}

	Depending on how star vectors are selected,
	it is possible to further simplify the code description.
	For instance, for the example code in \cref{sec:force},
	one can identify $ϕ↾¯x_1⊗S^3Y,ϕ↾¯x_2⊗S^3Y,ϕ↾¯x_3⊗S^3Y$
	with symmetric polynomials
	$𝘴_1(y,y',y''),𝘴_2(y,y',y''),𝘴_3(y,y',y'')∈𝔽[y,y',y'']_3$
	of tri-degree at most $(3,3,3)$ without quadratic terms.
	Here $¯x_1,¯x_2,¯x_3$ are a basis of $X$.
	Then the node content becomes the specialization
	$(𝘴_1+a_h^2𝘴_2+a_h^6𝘴_3)(a_h,y',y'')$.
	The help message becomes $(𝘴_1+a_h^2𝘴_2+a_h^6𝘴_3)(a_h,y',a_f)$.

讀
\section
								   Discussion   								

	\cite[section~IV]{RSKR09} argued that regenerating codes
	at the \MBR/ point must specialize to their proposal.
	Subsequently, the product-matrix construction at the \MBR/ point
	must be a direct generalization of the former proposal.
	This phenomenon is seen anew when
	determinant code generalizes layered code.
	The \cite{RSKR09}--product-matrix pair and
	the layered--determinant pair overlap at the $k=d$ \MBR/ point.
	
	On the other side, at the \MSR/ point,
	the product matrix is not succeeded by any more general code until now.
	We propose Atrahasis codes as a general code
	whose symmetric version includes the product matrix.
	Coincidentally, the exterior version of Atrahasis
	intersects the layered--determinant pair at the \MSR/ point.
	This inspires us to wonder whether
	the symmetric and exterior versions are in fact the same construction.
	
	Among \cref{tab:sageforce} there is a row highlighted.
	It is a $(14,10,12,84)$-\MSR/ code over the field of order $128$.
	It is, in particular, a $[14,10]$-\MDS/ code
	over the alphabet $𝔽_{128}^{84}$,
	which detects four errors or corrects two.
	It has parameters
	\[*\bma{
		1	&	d-k+1	&	α	&	α\log_2\abs{𝔽}	\\
		t	&	d		&	dβ	&	dβ\log_2\abs{𝔽}	\\
		k	&	*		&	M	&	M\log_2\abs{𝔽}	
	}=\bma{
		1	&	3		&	84	&	588				\\
		4	&	12		&	336	&	2352			\\
		10	&	*		&	840	&	5880			
	}.\]*
	For comparison, the improved Hadoop Distributed File System \cite{DD17}
	is a $[14,10]$-\MDS/ code over the field of order $256$.
	It has parameters
	\[*\bma{
		1	&	d-k+1	&	α	&	α\log_2\abs{𝔽}	\\
		*	&	d		&	*	&	dβ\log_2\abs{𝔽}	\\
		k	&	*		&	M	&	M\log_2\abs{𝔽}	
	}=\bma{
		1	&	4	&	1		&	8				\\
		*	&	13	&	*		&	54				\\
		10	&	*	&	10		&	80				
	}.\]*
	Note that the latter matrix is not of rank one
	because Hadoop is not an \MSR/ code to begin with.
	One sees that the $(14,10,12,84)$-Atrhasis has a huge \subpack/.
	But when it comes to homogeneous measures, such as $dβ/M$,
	our $2352/5880=40\%$ is much better than $54/80=67.5\%$.

\appendix

讀
\section
					   Bonus Property: Repairing Two Nodes at Once  					

\label{app:multiple}

	\begin{figure}
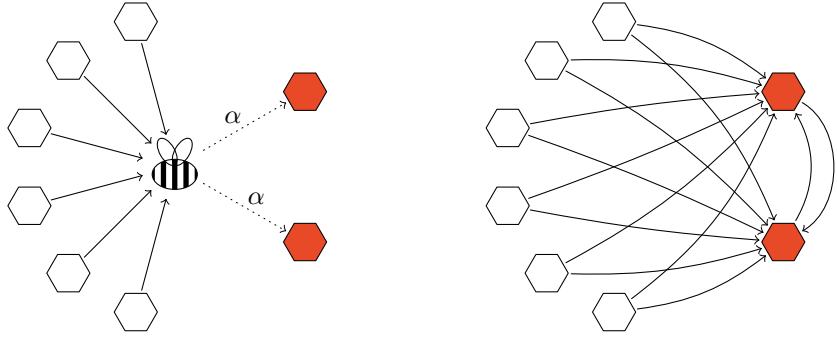

		$$\tikz[nodes={inner sep=1},shorten <=1,shorten >=1]{
			\begin{scope}
				\draw(-.1,.2)ellipse[x radius=.1,y radius=.2,rotate=30];
				\draw(.1,.2)ellipse[x radius=.1,y radius=.2,rotate=-30];
				\draw[clip](0,-.1)ellipse(.3 and .2);
				\draw[line width=2,scale=.15]
					foreach\x in{-3,...,3}{(\x,-2)--+(0,4)};
			\end{scope}
			\draw(0,0)node[circle,inner sep=8](bee){};
			\draw foreach\d in{1,...,6}{(\d*30+75:2)苯(D\d)(D\d)edge[->](bee)};
			\draw foreach\c in{1,...,2}{
				(\c*60-90:2)苯[fill=UCO](C\c)
				(bee)edge[->,dotted]node[auto]{$α$}(C\c)
			};
		}\hskip6em
		\tikz[nodes={inner sep=1},shorten <=1,shorten >=1]{
			\iflabor
			\draw foreach\d in{1,...,6}{
				(\d*30+75:2)苯(D\d)(D\d)\pgfextra{\PMS\ang{21-6*\d}}
				edge[->,bend left=\ang](C1)edge[->,bend left=\ang](C2)
			};
			\draw foreach\c in{1,2}{(\c*60-90:2)苯[fill=UCO](C\c)};
			\draw(C1)edge[->,bend left=-30](C2)(C2)edge[->,bend left=60](C1);
			\fi
		}$$
		\caption{
			To the left: centralized model.
			Healthy nodes send help messages to an agent.
			The total bandwidth $γ\ce(2,6)$ is
			the number of symbols passing solid lines.
			To the right: cooperative model.
			Healthy nodes send help messages directly to the failing nodes,
			while the latter can help each other.
			The total bandwidth $γ\co(2,6)$ is
			the number of symbols passing solid lines.
		}\label{fig:bee}
	\end{figure}
	
	In general, failures separate in time.
	But there may be circumstances where multiple nodes fail at once.
	\Cref{dfn:regenerate,dfn:msr} do not cover the case
	when there are more nodes to be repaired.
	What \cref{dfn:regenerate,dfn:msr} guarantee is that,
	A, so far as there are $k$ healthy nodes left, the file is safe.
	B, if there are $d$ healthy nodes left,
	one may call the repairing protocol for each and every failing node.
	This does not capture how efficient the repairing can be done.
	For that, two definitions are made in \cite{CJMRS13, SH13},
	and related in \cite{YB19}.
	
	\begin{dfn}\label{dfn:cooperative}
		Let there be $c$ failing nodes, and $d$ nodes are to help.
		The \emph{centralized (total) bandwidth} $γ\ce(c,d)$ is the total number
		of symbols the $d$ helper nodes send to a central agent,
		who will repair the failings after gathering all help messages.
		The \emph{cooperative (total) bandwidth} $γ\co(c,d)$ is
		how many symbols are sent over the network,
		from a helper node or a failing one, that contribute to repairing.
	\end{dfn}
	
	See \cref{fig:bee} for illustration.
	Note that we do not normalize the total bandwidth
	by the number of helping, or failing, nodes.
	One reason is that it is unclear what the denominator should be.
	When there is one failing node, $γ\ce(1,d)=γ\co(1,d)=dβ$.
	We now demonstrate how a $t=3$ Atrahasis code from \cref{sec:warmup}
	repairs two failing nodes.
	Recall the parameters $(d,α,β,M)=\(÷{3(k-1)}2,\bi{k-1}2,k-2,kα\)$
	and definitions $X≔𝔽^3$ and $Y≔𝔽^{k-2}$.

\subsection{The \PM t = 3$t=3$ Atrahasis code under centralized model}

	Say the $f$th and $g$th nodes fail and
	the first $d$ nodes will help through an agent (centralized model).
	To the $f$th node, the $h$th healthy node wants to send
	the restriction $ϕ↾x🔑h⊗y🔑h⊙Y⊙y🔑f$;
	to the $g$th it wants to send  $ϕ↾x🔑h⊗y🔑h⊙Y⊙y🔑g$.
	Since the agent will take care of the redistribution,
	the $h$th will simply send $ϕ↾x🔑h⊗y🔑h⊙Y⊙⟨y🔑f,y🔑g⟩$ to the agent.
	The dimension of the subspace $x🔑h⊗y🔑h⊙Y⊙⟨y🔑f,y🔑g⟩$,
	$2k-5$, is the number of symbols being sent.
	Multiplied by the number of helper nodes,
	\[*d(2k-5)=÷32(k-1)(2k-5)=3k^2-÷{21}2k+÷{15}2\]*
	is the total bandwidth $γ\ce(2,d)$ if done in this way.
	According to \cite{YB19}, the least possible bandwidth
	as a multiple of $α$ is
	\[*÷{2d}{d+2-k}·α=÷{3(k-1)}{3(k-1)/2+2-k}·÷{(k-1)(k-2)}2=3k^2-15k+O(1).\]*
	We are $9k/2$ away from the optimal value.
	Remark:
	\cite{YB19} achieve the optimal bandwidth $2dα/(d+2-k)$
	for a significantly greater $α$.
	
	In fact, Atrahasis code can do better.
	Imagine that when the agent receives the help messages,
	it will first reconstruct the $f$th node.
	Once done, the virtual $f$th node in agent's memory
	is able to send help message $ϕ↾x🔑f⊗y🔑f⊙Y⊙y🔑g$.
	In other words, one failing node becomes the helper of the other failing,
	so the second ($g$th) node requires $d-1$, not $d$, helpers.
	Therefore, the helper nodes should and will send the following messages:
	For $1≤h≤d-1$, the $h$th node sends $ϕ↾x🔑h⊗y🔑h⊙Y⊙⟨y🔑f,y🔑g⟩$
	to the agent, while the $d$th node sends $ϕ↾x🔑h⊗y🔑h⊙Y⊙y🔑f$.
	Then the reduced $γ\ce(2,d)$ is
	\[*(d-1)(2k-5)+1(k-2)=3k^2-÷{23}2k+÷{21}2.\]*
	
	In actuality, Atrahasis can do even better.
	Recall (in \cref{sec:warmup}) that
	the $f$th node learns $ϕ↾X⊗Y⊙Y⊙y🔑f$ from the help messages
	before it restricts to its usual content.
	For two failing nodes, the agent needs,
	and only needs, to learn $ϕ↾X⊗Y⊙Y⊙⟨y🔑f,y🔑g⟩$.
	This subspace has dimension $3(k-2)^2$.
	Regardless of which node should send what,
	we claim our final bandwidth
	\[*γ\ce(2,d)=3(k-2)^2=3k^2-12k+12.\]*
	This is $3k$ away from the optimality.

\section{Benchmarks}\label{app:benchmarks}

	In \crefrange{fig:alphagrid}{fig:EM19grid},
	we list some $α$ for small parameters from this and various other works.
	Compare them with lower bounds in \cref{fig:lowerbound}.
	See also \cref{tab:paradigm,tab:sageforce}.
	Note that Atrahasis's $α$ does not depend on $n$ while other works
	either stick to $n=d+1$ or have $α→∞$ as $n-d→∞$.
	In the following manner our \subpack/ level is polynomial
	in the lower bound (\cref{thm:lower}).
	
	\begin{thm}
		Let $t=d/(d-k+1)$.
		Let $α=\bi{k-1}{t-1}$.
		Then
		\[*α≤2^{k-1},(k-1)^{t-1}.\]*
		Moreover, if $d,k,t,α$ go to infinity with $d-k+1<C$ bounded, then
		\[*α<\exp（÷{k-1}{4(d-k+1)}）^{4C\log 2}.\]*
	\end{thm}
	
	\def\rmin{1}\def\rmax{5}\def\kmin{2}\def\kmax{17}
	\begin{figure}
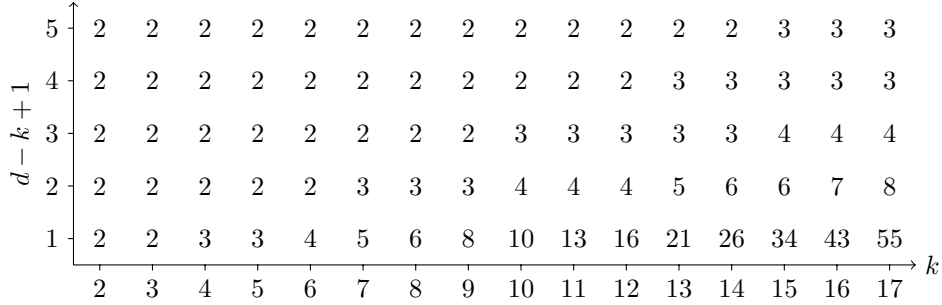

		$$\tikz[scale=.7]{
			\draw[<->]
				foreach\r in{\rmin,...,\rmax}{
					(1.5,\r)--+(-.1,0)node[left]{$\r$}
				}
				foreach\k in{\kmin,...,\kmax}{
					(\k,.5)--+(0,-.1)node[below]{$\k$}
				}
				(\kmin-1.5,\rmax/2+.5)node[rotate=90]{$d-k+1$}
				(\kmin-.5,\rmax+.5)|-(\kmax+.5,\rmin-.5)node[right]{$k$};
			\clip(\kmin-.5,\rmin-.5)rectangle(\kmax+.7,\rmax+.5);
			\iflabor
				\foreach\r in{1,...,\rmax}{
					\foreach\k in{2,...,\kmax}{
						\PMT\lowerbound{ceil(exp((\k-1)/4/\r))}
						\draw(\k,\r)node{$\lowerbound$};
					}
				}
			\fi
		}$$\caption{
			Lower bounds of the \subpack/ level $α$ from \cite[Theorem~1]{AG19}.
			The horizontal axis is $k$.
			The vertical axis is $d-k+1$ as $d<k$ cases are meaningless
			(and $d-(k-1)=d/t=(k-1)/(t-1)=α/β$
				bears more semantics than $d-k$ does).
			Omit $k=1$ for triviality.
			It is later discovered that their bound is not valid when $k=d$.
		}\label{fig:lowerbound}
	\end{figure}
	
	\begin{figure}
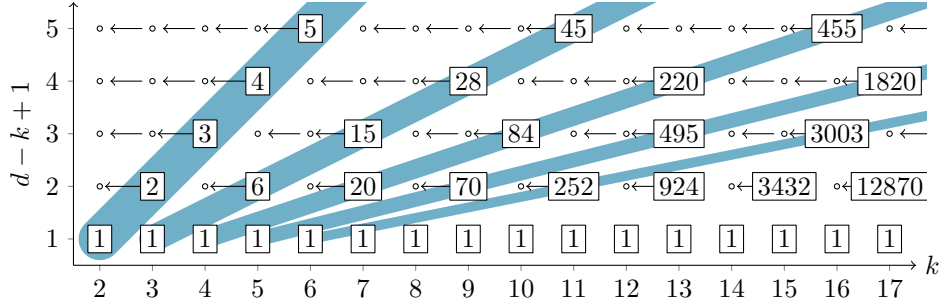

		$$\tikz[scale=.7]{
			\fill[white](\kmin-.5,\rmin-.5)rectangle(\kmax+.7,\rmax+.5);
			\draw[help lines,nodes=black]
				foreach\r in{\rmin,...,\rmax}{
					(1.5,\r)--+(-.1,0)node[left]{$\r$}
				}
				foreach\k in{\kmin,...,\kmax}{
					(\k,.5)--+(0,-.1)node[below]{$\k$}
				};
			\draw[<->]
				(\kmin-1.5,\rmax/2+.5)node[rotate=90]{$d-k+1$}
				(\kmin-.5,\rmax+.5)|-(\kmax+.5,\rmin-.5)node[right]{$k$};
			\clip(\kmin-.5,\rmin-.5)rectangle(\kmax+.7,\rmax+.5);
			\iflabor
				\begin{scope}[Gray-blue]
					\draw[line width=8mm/sqrt(2)](2,1)--(7,6);
					\draw[line width=8mm/sqrt(5)](3,1)--(13,6);
					\draw[line width=8mm/sqrt(10)](4,1)--(19,6);
					\draw[line width=8mm/sqrt(17)](5,1)--(21,5);
					\draw[line width=8mm/sqrt(37)](6,1)--(21,4);
				\end{scope}
				\foreach\r in{1,...,\rmax}{
					\foreach\k in{\kmax,...,2}{
						\PMT\t{ceil((\k-1)/\r+1)}	\PMT\rmndr{mod(\k-1+\r,\r)}
						\ifcase\rmndr
							\newcount\a\a1
							\ifnum\t>1
								\foreach\i in{1,...,\numexpr\t-1}{
									\multiply\a\numexpr\k-\i\global\divide\a\i
								}
							\fi
							\node(\k-\r)at(\k,\r)
								[inner sep=2,draw,fill=white]{$\the\a$};
						\else
							\PMT\kk{\k+1}
							\ifcsname pgf@sh@ns@\kk-\r\endcsname
								\draw
									(\k,\r)circle(.05)+(.1,0)edge[<-](\kk-\r);
							\else
								\draw(\k,\r)circle(.05)+(.2,0)edge[<-]+(.8,0);
							\fi
						\fi
					}
				}
			\fi
		}$$\caption{
			Upper bounds by Atrahasis code,
			our proposal of \MSR/ codes with arbitrary $(n,k,d)$.
			A box encloses the \subpack/ level $α$ of the primitive construction 
			at that point, which does not depend on $n$.
			Boxes on the same strip share the same $t$.
			From left to right $t=2,3…6$;
			the rest are omitted.
			One arrow means shortening once.
		}\label{fig:alphagrid}
	\end{figure}
	
	\begin{figure}
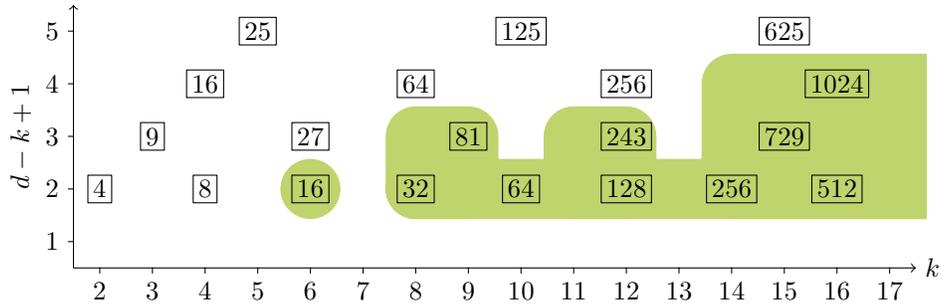

		$$\tikz[scale=.7]{
			\draw[<->]
				foreach\r in{\rmin,...,\rmax}{
					(1.5,\r)--+(-.1,0)node[left]{$\r$}
				}
				foreach\k in{\kmin,...,\kmax}{
					(\k,.5)--+(0,-.1)node[below]{$\k$}
				}
				(\kmin-1.5,\rmax/2+.5)node[rotate=90]{$d-k+1$}
				(\kmin-.5,\rmax+.5)|-(\kmax+.5,\rmin-.5)node[right]{$k$};
			\clip(\kmin-.5,\rmin-.5)rectangle(\kmax+.7,\rmax+.5);
			\iflabor
			\filldraw[Citron,line width=8mm](6,2)--+(0,.001)(8,2)--(18,2)
				(8,3)rectangle(9,2)(11,3)rectangle(12,2)(14,4)rectangle(18,2);
			\foreach\s in{2,...,9}{
				\foreach\q in{2,3,4,5,7,8,9}{
					\PMT\k{(\s-1)*\q}
					\ifnum\k>\kmax\else
						\PMT\n{\s*\q}	\PMT\d{\n-1}	\PMT\a{\q^\s}
						\draw(\k,\d-\k+1)node[inner sep=2,draw]{\a};
					\fi
				}
			}
			\fi
		}$$
		\caption{
			\cite{SAK15} gave \MSR/ codes with parameter quadruple
			$(n,k,d,α)=(sq,n-q,n-1,q^s)$ where $q$ is a prime power.
			Some $α$'s are put in boxes.
			Shortening applies but is omitted
			from this and the remaining figures.
			The shaded area is where their $α$ falls below ours.
		}\label{fig:SAK15grid}
	\end{figure}
	
	\begin{figure}
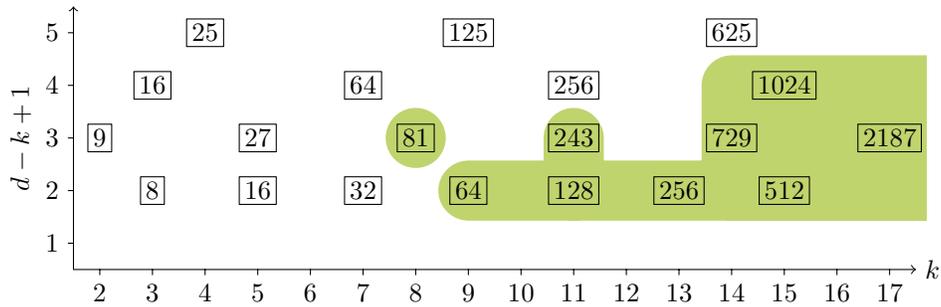

		$$\tikz[scale=.7]{
			\draw[<->]
				foreach\r in{\rmin,...,\rmax}{
					(1.5,\r)--+(-.1,0)node[left]{$\r$}
				}
				foreach\k in{\kmin,...,\kmax}{
					(\k,.5)--+(0,-.1)node[below]{$\k$}
				}
				(\kmin-1.5,\rmax/2+.5)node[rotate=90]{$d-k+1$}
				(\kmin-.5,\rmax+.5)|-(\kmax+.5,\rmin-.5)node[right]{$k$};
			\iflabor
			\clip(\kmin-.5,\rmin-.5)rectangle(\kmax+.7,\rmax+.5);
			\filldraw[Citron,line width=.8cm](8,3)--+(0,.001)(9,2)--(18,2)
				(11,3)--+(0,-1)(14,4)rectangle(18,2);
			\foreach\s in{2,...,9}{
				\foreach\q in{2,3,4,5,7,8,9}{
					\PMT\n{\s*\q}	\PMT\d{\n-2}	\PMT\k{\d-\q+1}
					\ifnum\k>\kmax\else
						\PMT\a{\q^\s}
						\draw(\k,\d-\k+1)node[inner sep=2,draw]{\a};
					\fi
				}
			}
			\fi
		}$$
		\caption{
			\cite{RKV16} extended \cite{SAK15}'s result
			via providing \MSR/ codes with parameters
			$(n,k,d,α)=(sq,n-q-m,n-1-m,q^s)$.
			In other words, they allows $n>d+1$.
			We plot the $α$ when $m=1$ (i.e., when $n=d+2$);
			and shade area where their $α$ falls below ours.
			Note that for any fixed $k,d$ such that $k<d$,
			their $α$ exceeds ours as $n-d$ increases.
		}\label{fig:RKV16grid}
	\end{figure}
	
	\begin{figure}
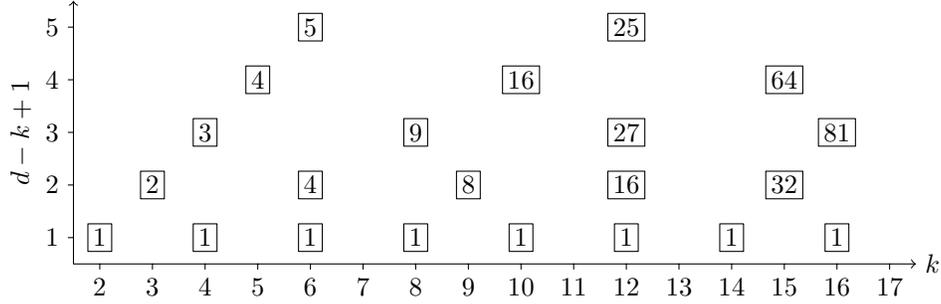

		$$\tikz[scale=.7]{
			\draw[<->]
				foreach\r in{\rmin,...,\rmax}{
					(1.5,\r)--+(-.1,0)node[left]{$\r$}
				}
				foreach\k in{\kmin,...,\kmax}{
					(\k,.5)--+(0,-.1)node[below]{$\k$}
				}
				(\kmin-1.5,\rmax/2+.5)node[rotate=90]{$d-k+1$}
				(\kmin-.5,\rmax+.5)|-(\kmax+.5,\rmin-.5)node[right]{$k$};
			\clip(\kmin-.5,\rmin-.5)rectangle(\kmax+.7,\rmax+.5);
			\iflabor
				\foreach\m in{1,...,9}{
					\foreach\r in{1,...,9}{
						\PMT\k{(\r+1)*\m}
						\ifnum\k>\kmax\else
							\PMT\n{\k+\r}	\PMT\d{\n-1}	\PMT\a{\r^\m}
							\draw(\k,\d-\k+1)node[inner sep=2,draw]{\a};
						\fi
					}
				}
			\fi
		}$$
		\caption{
			\cite{WTB16} contributed \MSR/ codes with parameters
			$(n,k,d,α)=(k+r,(r+1)m,n-1,r^m)$.
			They only enforce the optimal repair bandwidth for systematic nodes.
			This results in the least possible $α$
			among all \MSR/-related codes we have seen.
		}\label{fig:WTB16grid}
	\end{figure}
	
	\begin{figure}
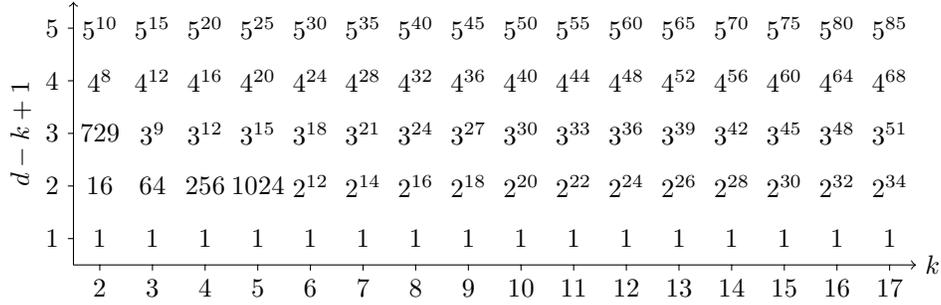

		$$\tikz[scale=.7]{
			\draw[<->]
				foreach\r in{\rmin,...,\rmax}{
					(1.5,\r)--+(-.1,0)node[left]{$\r$}
				}
				foreach\k in{\kmin,...,\kmax}{
					(\k,.5)--+(0,-.1)node[below]{$\k$}
				}
				(\kmin-1.5,\rmax/2+.5)node[rotate=90]{$d-k+1$}
				(\kmin-.5,\rmax+.5)|-(\kmax+.5,\rmin-.5)node[right]{$k$};
			\clip(\kmin-.5,\rmin-.5)rectangle(\kmax+.7,\rmax+.5);
			\iflabor
				\foreach\r in{\rmin,...,\rmax}{
					\foreach\k in{\kmin,...,\kmax}{
						\PMT\rr{\r+1}	\PMT\ro{\r}
						\PMT\kr{\k*\r)}
						\PMS\loga{\kr*log10(\ro)}
						\ifdim\loga pt<3.1pt
							\PMT\a{\ro^(\kr)}
						\else
							\def\a{\ro^{\kr}}
						\fi
						\draw(\k,\r)node{$\a$};
					}
				}
			\fi
		}$$
		\caption{
			\cite{GFV17} gave codes with parameters
			$(n,k,d,α)=\(k+r,k,k+ρ+1,ρ^{k\bi r{ρ}}\)$.
			They only enforce the optimal repair bandwidth for systematic nodes.
			Note that their $α$ depends on $n$ beyond $k,d$.
			When $n=d+1$, it coincides with \cref{fig:SAK15grid}.
			We display the $n=d+2$ case here.
			It was remarked that their $α$ could be optimized further
			but we decided to print the very $α$ given therein.
			Note that \cite{GFV17} is chronologically before \cite{RKV16}.
		}\label{fig:GFV17grid}
	\end{figure}
	
	\begin{figure}
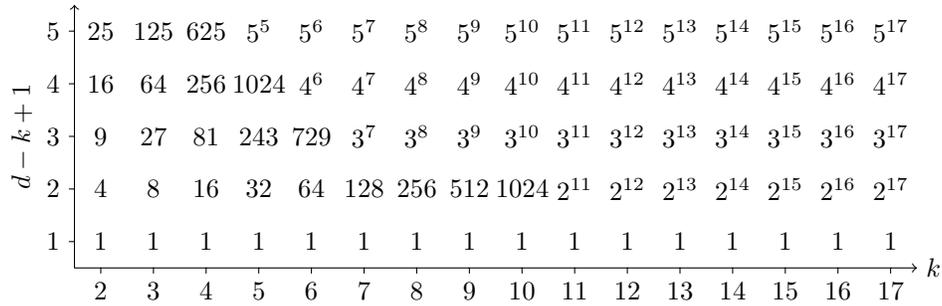

		$$\tikz[scale=.7]{
			\draw[<->]
				foreach\r in{\rmin,...,\rmax}{
					(1.5,\r)--+(-.1,0)node[left]{$\r$}
				}
				foreach\k in{\kmin,...,\kmax}{
					(\k,.5)--+(0,-.1)node[below]{$\k$}
				}
				(\kmin-1.5,\rmax/2+.5)node[rotate=90]{$d-k+1$}
				(\kmin-.5,\rmax+.5)|-(\kmax+.5,\rmin-.5)node[right]{$k$};
			\clip(\kmin-.5,\rmin-.5)rectangle(\kmax+.7,\rmax+.5);
			\iflabor
				\foreach\r in{\rmin,...,\rmax}{
					\foreach\k in{\kmin,...,\kmax}{
						\PMS\loga{\k*log10(\r)}
						\ifdim\loga pt<3.1pt
							\PMT\a{\r^(\k)}
						\else
							\def\a{\r^{\k}}
						\fi
						\draw(\k,\r)node{$\a$};
					}
				}
			\fi
		}$$
		\caption{
			\cite{EM19c} provided regenerating codes
			with arbitrary parameters $(n,k,d)$.
			These include \MSR/ codes, whose $α$ are shown above,
			with $(n,k,d,α)=(n,k,d,(d-k+1)^k)$.
			Their $α$ is higher than ours in general.
			The same parameters can be achieved using other constructions;
			see \cite{DL19,DLW20}.
		}\label{fig:EM19grid}
	\end{figure}

\hbadness9999
\makeatletter
\g@addto@macro\sloppy{\advance\baselineskip0ptplus1ptminus1pt}
\let\($\let\)$
\bibliographystyle{alphaurl}
\bibliography{PowerCrystal-3}

\end{document}